\documentclass[journal]{IEEEtran}

\if CLASSOPTIONcompsoc
  \usepackage[nocompress]{cite}
\else
  \usepackage{cite}
\fi

\usepackage[utf8x]{inputenc}
\usepackage{listings}
\usepackage{color}
\usepackage{epsfig}
\usepackage{url}
\usepackage{amsmath}
\usepackage{cite}
\usepackage{graphicx}
\usepackage{color}
\usepackage[bookmarks=false]{hyperref}
\usepackage{amssymb,amsmath,dsfont}
\usepackage{amsthm}
\usepackage[linesnumbered,lined, ruled]{algorithm2e}
\usepackage{epstopdf}
\usepackage{lipsum}
\usepackage{bbm}
\usepackage{tablefootnote}

\usepackage{tikz}
\usetikzlibrary{automata,arrows,positioning,calc}

\usepackage{multirow}

\usepackage{breqn}
\SetAlFnt{\footnotesize}

\usepackage{graphics}

\newtheorem{theorem}{Theorem}

\newtheorem{proposition}{Proposition}
\newtheorem{lemma}{Lemma}

\theoremstyle{remark}

\theoremstyle{application}
\newtheorem*{application}{Example}

\usepackage{csquotes}
\usepackage{amsmath} 
\usepackage{amssymb}
\usepackage{subcaption}
\usepackage{graphicx}
\usepackage{accents}
\usepackage{mathtools}
\usepackage[normalem]{ulem}

\SetCommentSty{mycommfont}

\newcommand{\R}{\mathbb{R}}
\newcommand{\N}{\mathbb{N}}

\newcommand{\calA}{\mathcal{A}}
\newcommand{\calD}{\mathcal{D}}
\newcommand{\Finc}{\mathcal{F}_{\mathrm{inc}}}
\newcommand{\Fincz}{\mathcal{F}^{0}_{\mathrm{inc}}}
\newcommand*{\eqdef}{\stackrel{\mathrm{def}}{=}}
\newcommand*{\alphapcr}[1]{
    \ifthenelse{\isempty{#1}}
    {\alpha_{\mathrm{pkt}}}
    {\alpha_{\mathrm{pkt},#1}}
}
\newcommand{\A}{\mathrm{A}}
\newcommand{\G}{\mathrm{G}}
\newcommand{\pkt}{\mathrm{pkt}}
\newcommand{\us}{~\mu\mathrm{s}}

\newcommand{\mbps}{~\mathrm{Mbps}}

\newcommand{\mfa}{\mbox{ for all }}
\newcommand{\mand}{\mbox{ and }}

\graphicspath{figures/}

\usepackage{mathtools}
\DeclarePairedDelimiter\floor{\lfloor}{\rfloor}
\DeclarePairedDelimiter\ceil{\lceil}{\rceil}

\SetKwRepeat{Repeat}{repeat}{until}
\SetKwRepeat{Forever}{repeat}{forever}
\SetKwRepeat{On}{on}{end}
\SetNlSty{bfseries}{\color{black}}{}

\begin{document}

\title{Improved Network Calculus Delay Bounds in Time-Sensitive Networks}


\author{\IEEEauthorblockN{
		Ehsan Mohammadpour, Eleni Stai, Jean-Yves Le Boudec
		\\}
\IEEEauthorblockA{\'Ecole Polytechnique F\'ed\'erale de Lausanne, Switzerland\\
$\{$firstname.lastname$\}$@epfl.ch}}

\maketitle

\begin{abstract}
In time-sensitive networks, bounds on worst-case delays are typically obtained by using network calculus and assuming that flows are constrained by bit-level arrival curves. However, in IEEE TSN or IETF DetNet, source flows are constrained on the number of packets rather than bits. A common approach to obtain a delay bound is to derive a bit-level arrival curve from a packet-level arrival curve. However, such a method is not tight: we show that better bounds can be obtained by directly exploiting the arrival curves expressed at the packet level. Our analysis method also obtains better bounds when flows are constrained with g-regulation, such as the recently proposed Length-Rate Quotient rule. It can also be used to generalize some recently proposed network-calculus delay-bounds for a service curve element with known transmission rate.
\end{abstract}


\IEEEpeerreviewmaketitle

\section{Introduction}\label{sec:intro}
%
Time-sensitive networks provide real-time guarantees for applications such as avionics, automobiles, industrial automation, etc \cite{ieeeAVB,iecIECIEEE608022019,ieeeDraftStandardLocal2019b,ecssSpaceWireLinksNodes2008,tsn-profile-service-provider,itu-y3000}. IETF Deterministic Networking (DetNet) \cite{detnet} and IEEE Time-Sensitive Networking (TSN) \cite{tsn} formalize the requirements and provide standardization for such networks. One of the main goals in time-sensitive networks is to provide guarantees on worst-case delay, and not average delay. In order to obtain such guarantees, flows are assumed to be regulated at the sources. A classical form of source regulation is the bit-level arrival curve constraint: a flow satisfies a bit-level arrival-curve constraint $\alpha()$ if the number of bits over any time interval of any duration $t$ is upper bounded by $\alpha(t)$ \cite{le_boudec_network_2001,bouillard_deterministic_2018}. The celebrated token bucket constraint with rate $r$ and burst $b$ is an example, with $\alpha(t)=rt+b$ for $t>0$. Formally proven delay bounds can be obtained by using network calculus, which combines arrival and service curves \cite{le_boudec_network_2001,bouillard_deterministic_2018}. A service curve is an abstraction that conveys information on the minimum service provided by a system (see Section~\ref{sec:sys}).  Then the \textit{classical} network-calculus delay-bound is obtained by taking the horizontal deviation of the arrival and service curves \cite[Section 3]{le_boudec_network_2001} \cite[Section 3.1.2]{bouillard_deterministic_2018}. This delay bound can be improved for some systems with known transmission rate~\cite{mohammadpour_improved_2019}. 

However, in time-sensitive networks, flow regulation at sources is often expressed in terms of number of packets rather than number of bits \cite{rfc9016} \cite[Section 35.2.2.8.4]{ieee8021q}; for example in TSN, the number of packets observed within any fixed \textit{class measurement interval} (CMI) is upper-bounded by a constant value. To obtain delay bounds for such flows, network calculus is often used \cite{boundedlatency}, for which a bit-level characterization is required. Hence, a common approach is to derive a bit-level arrival-curve from the packet-level constraint and then apply network calculus~\cite{daigmorte_modelling,maile2020network}. 
However, as we show in Table~\ref{table:bound-compare} and Section~\ref{sec:eval}, this can lead to delay bounds that can be improved. 
Indeed, our main result, in Theorem \ref{thm:delay_pcr}, is
a novel delay bound for flows regulated with packet-level constraints, which improves on the one that is obtained by deriving a bit-level arrival-curve from the packet-level regulation constraint. We show that the obtained delay bound is tight at least for $c$-Lipschitz \cite{eriksson2004applied} service curves where $c$ is the physical line-speed (Theorems~\ref{thm:delay_pcr_tight} and~\ref{thm:delay_tight_fixed_interval}). 

Our method uses a novel modelling of packet-level constraint with g-regulation 
(Proposition \ref{pro:rate2reg}). The concept of g-regulation was introduced by C.S. Chang \cite{changbook} as an alternative to bit-level arrival-curve and uses max-plus algebra, whereas results for bit-level arrival-curves tend to use min-plus algebra. Our improvement  in the delay bound is made possible by combining min-plus representation of service curves and max-plus representation of the input traffic to obtain a bound on queuing delay (Lemma~\ref{lem:queuing-delay}).

 As a by-product of our method of proof, we also obtain delay bounds for flows with g-regulation, such as Length-Rate Quotient (LRQ) \cite{specht_urgency_based_2016}, that share a FIFO system with known service curve (Theorem~\ref{thm:greg-delay}). This bound improves on the one obtained by deriving a bit-level arrival-curve from g-regulation and then using the state-of-the-art network calculus bound. Moreover, it generalizes the results by \cite{specht_urgency_based_2016} that was specifically applicable to flows with LRQ constraints and priority queues with constant rate servers. Last, the state-of-the-art network-calculus delay-bound for bit-level arrival-curve can also be generalized to a wider family of service curves than rate-latency ones (Theorem~\ref{thm:response_fifo_arrival}). 

The rest of the paper is organized as follows. Section~\ref{sec:related} presents the state-of-the-art and related works. Section~\ref{sec:sys} includes the system model, notation, and the definition of the considered flow regulation constraints.
Section~\ref{sec:regModel} presents the relations among different regulation constraints.
In Section~\ref{sec:delay}, we present the main contribution of this paper, namely, novel delay bounds for flows with packet-level constraints, and we show that these bounds are tight. Section~\ref{sec:ncimprov} gives a generalization of the existing delay bounds for flows with g-regularity constraint and bit-level arrival-curve, as a by-product of our method of proof in Section~\ref{sec:delay}. Section~\ref{sec:comparison} shows that the best delay bound
for a given flow is obtained by directly applying the theorem
corresponding to its initial constraint.
Section~\ref{sec:eval} provides a numerical illustration of the theorems presented in this paper and Section~\ref{sec:conclusion} concludes the paper. The Appendix gives proofs.

\section{Related Works}\label{sec:related}
The classical network-calculus proves delay bounds for a FIFO system with bit-level arrival and service curves \cite{le_boudec_network_2001,bouillard_deterministic_2018}. Specifically, for a flow with bit-level arrival-curve $\alpha$ that enters a FIFO system with service curve $\beta$, the bound is obtained by taking the horizontal deviation between the two curves, i.e., $h(\alpha,\beta)$, which is defined in \cite[Section~3.1.11]{le_boudec_network_2001} and recalled in Section~\ref{sec:sys}. For example, when a flow has leaky bucket arrival curve $\alpha(t)=rt+b$ and the FIFO system has rate-latency service-curve $\beta(t)=\max\left(R(t-T),0\right)$, we have $h(\alpha,\beta)=T+\frac{b}{R}$ provided that $r\leq R$. This bound is tight if only arrival and service curves are known. However, we often have more information, specially in the context of time-sensitive networks, e.g., TSN schedulers with Credit-Based Shapers \cite{mohammadpour_latency_2018} and Deficit Round Robin (DRR) schedulers \cite{tabatabaee2021deficit}: when a packet starts its transmission, it is transmitted with full line rate. Recently, such information was exploited for rate-latency service-curves to provide a delay bound that improves the classical network-calculus bound \cite{mohammadpour_improved_2019}; more precisely, the classical network-calculus bound is reduced by $L^{\min}\left(\frac{1}{R}-\frac{1}{c}\right)$, where $c$ is the physical line-rate and $L^{\min}$ is the minimum packet length of the input traffic.

While rate-latency service-curves are commonly used in the literature, they may not provide a perfect characterization of a system. A number of works provide more complex service curves that in turn leads to better delay bounds. In \cite{tabatabaee2021deficit}, the authors obtained a non rate-latency service-curve for Deficit Round-Robin (DRR) schedulers that can incorporate the arrival curves of the interfering flows. Such non rate-latency service-curves are obtained as well for Weighted Round-Robin (WRR)\cite[Section 8.2.4]{bouillard_deterministic_2018} and Interleaved WRR \cite{tabatabaee2021interleaved}. Hence, the improvement of \cite{mohammadpour_improved_2019} does not apply to these service curves. In this paper, we generalize the results of \cite{mohammadpour_improved_2019} to improve delay bounds for non rate-latency service-curves. 
Furthermore, we improve the bounds in \cite{mohammadpour_improved_2019} when the flows have packet-level constraints.


A number of other works focus on improving the arrival curves of the flows in time-sensitive networks by taking advantage of input-line shaping effect. For example, for a flow with leaky-bucket arrival curve $\alpha(t)=rt+b$ that passes a physical line with rate $c$, we can obtain a better arrival curve $\alpha'(t)=\min(rt+b,ct)$. In \cite{daigmorte_modelling,maile2020network}, the authors study a TSN network and assume the input traffic has packet-level constraint \cite[Section 35.2.2.8.4]{ieee8021q}. To obtain delay bounds, they first derive a bit-level arrival-curve from the packet-level constraint; then they exploit input-line shaping to improve the obtained arrival curve. Finally, they use network calculus to obtain delay bounds. 
The improvement by the input-line shaping effect is complementary to the improvements in this paper. 

Recently, LRQ was introduced in the context of interleaved regulators in \cite{specht_urgency_based_2016} as a per-flow regulation constraint that is simpler to implement than token-bucket. LRQ is in fact a specific case of Shifted-Rate Regulator (when the delay term is set to zero) \cite[Section~6.2.1]{changbook}, which belongs to the family of g-regulation. \cite{specht_urgency_based_2016} obtains delay bounds for flows with LRQ constraint within constant-rate servers and strict-priority queuing. This analysis was extended for guaranteed-rate servers \cite{jiang2021properties}. Unlike the previous works, our results on g-regulated flows cover the whole family of g-regulation (i.e., are not limited to LRQ constraint) and apply to a broader class of network nodes (i.e., nodes with arbitrary service curves).

\section{System model and General Prerequisite}\label{sec:sys}

\begin{table}[t]
\renewcommand{\arraystretch}{1.3} 
\caption{List of notation.}
\label{table:notation}
\begin{tabular}{|c|p{0.8\columnwidth}|}
\hline
Term          & Description                                                       \\ \hline \hline
$A_n$         & The arrival time of packet $n$ to the FIFO system                 \\ \hline
$c$           & The transmission rate of the output link                          \\ \hline
$D_n$         & The departure time of packet $n$ from the FIFO system             \\ \hline
$h(w',w)$      & The horizontal deviation from function $w'$ to function $w$ (Section~\ref{sec:notations})       \\ \hline 
$K$           & The maximum  number of packets in each interval of TSN/DetNet traffic constraint        \\ \hline
$L^\max_f$    & Maximum packet length of flow $f$                                 \\ \hline
$L^\min_f$    & Minimum packet length of flow $f$                                 \\ \hline
$l_n$         & The length of packet $n$ in bits                          \\ \hline
$N(t)$        & The cumulative number of packets arrived at the FIFO system until time $t$~(excluded)       \\ \hline 
$Q_n$         & The start of transmission of packet $n$ from the FIFO system      \\ \hline
$R$           & Rate of service curve, when $\beta$ is rate-latency               \\ \hline
$T$           & Latency of service curve, when $\beta$ is rate-latency            \\ \hline
$\beta$       & The service curve of the FIFO system                              \\ \hline
$\Delta^\A$   & The response time bound on a flow with bit-level arrival-curve    \\ \hline
$\Delta^{\A\G}$   & The response time bound on a flow with $g$-regularity constraint competing with an aggregate of flows with bit-level arrival-curve   \\ \hline
$\Delta^\G$   & The response time bound on a flow with $g$-regularity constraint  \\ \hline
$\Delta^\pkt$ & The response time bound on a flow with packet-level arrival-curve \\ \hline 
$\tau$        & The interval duration in TSN/DetNet traffic constraint     \\ \hline \hline
$\N$          & The set of positive integers, i.e., $\{1,2,\dots,\}$        \\ \hline 
$\R^+$        & The set of non-negative real numbers, i.e., $[0,\infty)$          \\ \hline 
$\Fincz$      & The set of wide-sense increasing functions such that $\forall w\in \Fincz:~w(0)=0$        \\ \hline 
$w^+$       & The right limit of function $w$ (Section~\ref{sec:notations}         \\ \hline 
$w^-$       & The left limit of function $w$ (Section~\ref{sec:notations}         \\ \hline 
$w^{\downarrow}$       & The lower pseudo-inverse of function $w$ (Section~\ref{sec:preq})              \\ \hline 
$w^{\uparrow}$       & The upper pseudo-inverse of function $w$  (Section~\ref{sec:preq})                \\ \hline 
$\ceil{x}$    & The ceiling of $x$                                                \\ \hline 
$\floor{x}$    & The floor of $x$                                                \\ \hline 
$[x]^+$       & $\max\{x,0\}$          \\ \hline
$\circ$   & The function-composition operator, i.e., $(f \circ g )(x)=f(g(x))$ \\ \hline 
$1_{\{C\}}$   & It is equal to $1$ when condition $C$ is true; otherwise it is $0$. \\ \hline 
\end{tabular}
\end{table}

In this section, we first describe the systems to which our analysis applies. Second, we recall the definition of pseudo-inverse functions as they are used throughout the paper. Third, we give a description of the flow regulation constraints considered in this paper, i.e., bit-level arrival-curve, g-regulation and packet-level arrival-curve. Finally, we give a summary of necessary mathematical definitions required to follow the content of the paper. For the reader's convenience, Table~\ref{table:notation} gives the notation used throughout the paper.
\subsection{System Model}
\begin{figure}[t]
	\centering
	\resizebox{0.6\linewidth}{!}{\begin{tikzpicture}
\tikzstyle{roundedbox} =[draw,rounded corners=0.1cm];
\tikzstyle{vroundedbox} =[draw,rounded corners=0.1cm,rotate=90];
\tikzstyle{general} = [draw,line width=1pt]

\def \IRlen{1.75cm} 
\def \IRwid{0.5cm} 
\def \FIFOlen{3cm} 
\def \FIFOwid{1cm} 
\coordinate (picLeft) at (0,0);
\coordinate (IRLeft) at (picLeft);

\def \IRdis{0.1cm} 



\draw[general,->,anchor=west,text width=1cm] ([xshift = -0.5cm]IRLeft) -- ([xshift = 0.1cm]IRLeft) node [pos=-1] (input){Input traffic};

\coordinate (FIFOLeft) at ([xshift=0.2cm]picLeft); 
\draw [general] ([yshift = \FIFOwid/2]FIFOLeft) -- ([xshift=\FIFOlen,yshift = \FIFOwid/2]FIFOLeft) -- ([xshift=\FIFOlen,yshift = -1* \FIFOwid/2]FIFOLeft) node [pos=0.5](FIFOEnd){} -- ([yshift = -1*\FIFOwid/2]FIFOLeft); 
\node[anchor=east] at ([xshift = -\IRlen/2]FIFOEnd.center) { FIFO};

\coordinate (arrowInFIFO) at ([xshift = 0.2cm]FIFOLeft.center); 

\node[anchor=south] at ([xshift = -\FIFOlen/2, yshift = \FIFOwid/2+0.3cm]FIFOEnd.center) (HPFIFO){\fontsize{8}{48} \selectfont Other queues};
\coordinate (FIFOTop) at ([xshift = -\FIFOlen/2,yshift = \FIFOwid/2+0.1cm]FIFOEnd.center);
\path (HPFIFO) -- (FIFOTop) node [font=\small, midway, sloped] {$\dots$};

\node[anchor=north] at ([xshift = -\FIFOlen/2, yshift = -\FIFOwid/2-0.3cm]FIFOEnd.center) (LPFIFO){\fontsize{8}{48} \selectfont Other queues};
\coordinate (FIFODown) at ([xshift = -\FIFOlen/2,yshift = -\FIFOwid/2-0.1cm]FIFOEnd.center);
\path (LPFIFO) -- (FIFODown) node [font=\small, midway, sloped] {$\dots$};

\coordinate (schedulerLeft) at ([xshift=0cm]FIFOEnd); 
\node [general,rotate=90 , anchor=north] at ([xshift=-0.75pt]schedulerLeft.center)(scheduler){Scheduler} ; 

\def \transRadius{0.3cm} 
\coordinate (transLeft) at ([xshift=0cm]scheduler.south); 
\node[general,circle,minimum height=\transRadius*2, radius=\transRadius,anchor=west] at ([xshift=.1pt]transLeft.center)(trans){c}; 

\def \transArrow{0.3cm} 
\draw[general,->,anchor=west] (trans.east) -- ([xshift = \transArrow]trans.east);

\draw[general,<->,anchor=south] ([yshift = 1.5cm]FIFOLeft.center) -- ([yshift = 1.5cm]trans.east) node [pos=0.5]{Service curve $\beta$};

\node [general,dotted, anchor=west,minimum height=\FIFOwid] at (arrowInFIFO.center)(packetnFIFOIn){$l_n$} ; 


\node [general, dotted, anchor=west,minimum height=\IRwid] at ([yshift=\transRadius+0.5cm]trans.west)(packetnHead){$l_n$} ; 

\node [general, anchor=west,minimum height=\IRwid] at ([yshift=\transRadius+0.2cm]trans.east)(packetnTrans){$l_n$} ; 


\draw[general,-,dotted]([yshift=-\IRwid/2-\IRdis-\IRwid-0.2cm]packetnFIFOIn.west) -- ([yshift=0.5cm]packetnFIFOIn.west) node [pos=0,yshift=-0.2cm]{$A_n$};

\draw[general,-,dotted]([yshift=-\IRwid/2-\IRdis-\IRwid-\FIFOwid/2-0.5cm]packetnHead.west) -- ([yshift=0.5cm]packetnHead.west) node [pos=0,yshift=-0.2cm]{$Q_n$};

\draw[general,-,dotted]([yshift=-\IRwid/2-\IRdis-\IRwid-\FIFOwid/2-0.2cm]packetnTrans.west) -- ([yshift=0.5cm]packetnTrans.west) node [pos=0,yshift=-0.2cm]{$D_n$};

\end{tikzpicture}}
	\caption{The considered FIFO system.}
	\label{fig:fifo_node}
\end{figure}
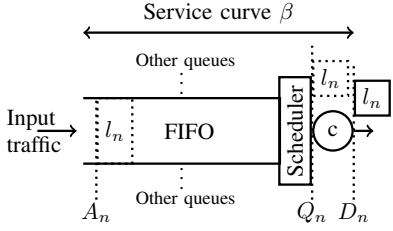 

We study a FIFO system with a queue and a transmission link, as in Fig. \ref{fig:fifo_node}. Each queue is shared among a number of flows. Upon arrival, packets of different flows enter the queue and are stored in FIFO order. A scheduler decides when the packet at the head of the queue is selected for transmission. The scheduler typically arbitrates between this queue and other queues (not shown), therefore the packet at the head of the queue may have to wait even if there is no packet of this queue in transmission. When the packet at the head of this queue is selected for transmission, it is transmitted at a constant rate $c$ until it is completely transmitted, i.e., there is no preemption. 

Every flow is assumed to be constrained by either bit-level arrival curve (Section~\ref{sec:arr}), g-regulation (Section~\ref{sec:greg}) or packet-level arrival curve (Section~\ref{sec:pktarr}).
For each flow $f$, let $L^{\min}_f$ and $L^{\max}_f$ denote the minimum and maximum packet lengths, in bits. We also assume that the total flow of all incoming packets is packetized, i.e., we consider that all bits of any packet arrive at the same time instant.

Let $A_n$ be the arrival time of packet $n$ to the FIFO system, where the numbering of packets is by order of arrival and $Q_n$ be the time at which packet $n$ is selected for transmission. The FIFO assumption means that $Q_n \leq Q_{n+1}$. We call $Q_n-A_n$ the queuing delay of packet $n$ in the FIFO system. 
Let $l_n$ be the length in bits of packet $n$ and assume that it belongs to flow $f$, so that $L^{\min}_f\leq l_n\leq L^{\max}_f$; then packet $n$ leaves the system at time $D_n = Q_n +\frac{l_n}{c}$. We call $D_n-A_n$ the response time of the FIFO system for packet $n$. 

We assume that the scheduler is such that the FIFO system offers to the total flow of all incoming packets a service curve $\beta$. Formally, this is defined as follows \cite[Section 1.3]{le_boudec_network_2001}. Let $\Fincz$ be the set of wide-sense increasing functions $w: [0,+\infty)\to [0,+\infty]$ such that  $w(0)=0$. Let 
$\calA, \calD \in \Fincz$ be such that $\calA(t)$ [resp. $\calD(t)$] denotes the cumulative number of bits arrived in [resp. departed from] the system until time $t$ (excluded). A function $\beta\in \Fincz$ is a service curve for the system if for every time $t\geq 0$ there exists a time $s\leq t$ such that 
\begin{equation}\label{eq:arrival_def}
\calD(t) \geq \beta(t-s)+ \calA(s).
\end{equation}
A service curve characterization is available for many systems, see for example \cite{zhao_timing_2018,boyer2012drr,finzi2018bls,mangoua2012sp,burchard2018gps}
Rate-latency service-curves are functions of the form  \mbox{$\beta(t)=R[t-T]^+$} with $R,T$ being the rate and latency terms. A system that offers a rate-latency service curve can be interpreted as behaving, for the flows of interest, as if it would be a server with rate $R$ and vacation $T$. The rate $R$ is the rate guaranteed to the flow and is typically less than the line rate $c$. Rate-latency service curves are often used because of their simplicity, but better delay bounds can also be obtained with more complex service curves \cite{tabatabaee2021deficit}.

\subsection{Pseudo-inverse Functions}\label{sec:preq}

For $w\in \Fincz$, the lower and upper pseudo-inverses are respectively $w^{\downarrow}$ and $ w^{\uparrow}$, and are defined as:
\begin{align}
    \label{eq:lsi-def} w^{\downarrow}(x) &\eqdef\inf \{s\geq 0~|~w(s)\geq x\},\\
    \label{eq:usi-def} w^{\uparrow}(x) &\eqdef\sup \{s\geq 0~|~w(s)\leq x\}.
\end{align}
Fig. \ref{fig:inverse} illustrates the pseudo-inverse functions and the differences between them.
By \cite[Section 10.1]{liebeherr_duality_2017}:
\begin{itemize}
    \item $w^{\downarrow}$ is non-decreasing and left continuous.
    \item $w^{\uparrow}$ is non-decreasing and right continuous.
\end{itemize}
\begin{figure}[h]
    \centering
    \includegraphics[width=0.9 \linewidth]{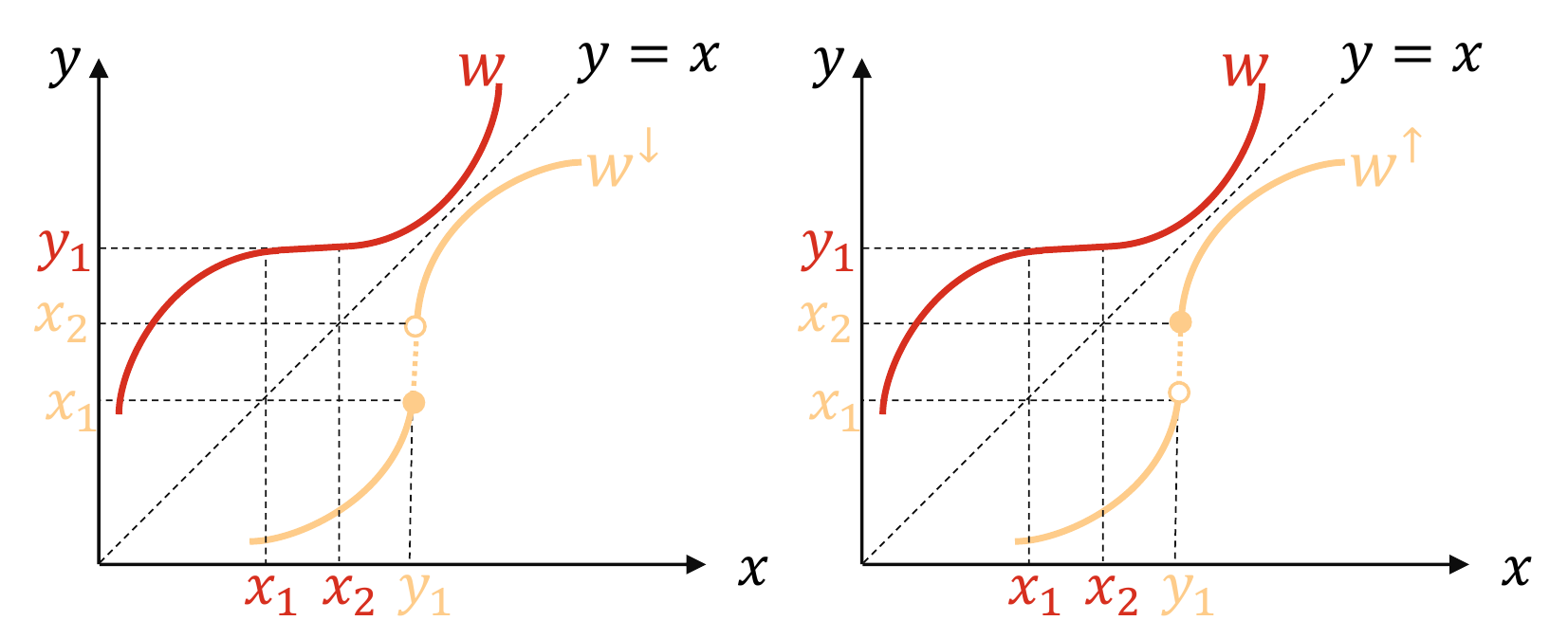}
    \caption{Illustration of pseudo-inverse functions of a monotonically increasing function $w$. The pseudo-inverse functions are obtained by flipping the graph of $w$ around the line $y = x$. The resulting graph does not correspond to a function as the plateau part of $w$ i.e., $x_1$ to $x_2$, causes ambiguity. With the lower pseudo-inverse (left figure), $w^\downarrow$, the ambiguity is resolved by selecting the infimum, i.e., $w^\downarrow(y_1) = x_1$. With the upper pseudo-inverse (right figure), $w^\uparrow$, the ambiguity is resolved by selecting the supremum, i.e., $w^\uparrow(y_1) = x_2$. }
    \label{fig:inverse}
\end{figure} 


\subsection{Bit-level Arrival-Curve}\label{sec:arr}
The bit-level arrival-curve constraint is the most widespread form of traffic regulation \cite{le_boudec_network_2001}.
Consider a left-continuous function $\calA \in \Fincz$, where $\calA(t)$ denotes the cumulative number of bits observed on the flow of interest until time $t$ (excluded). We say that the flow is constrained by the bit-level arrival-curve $\alpha \in \Fincz$  if and only if   
\begin{equation}\label{eq:arrival_def}
\calA(t)-\calA(s) \leq \alpha(t-s)\; \mfa t\geq 0 \mand s\in[0,t].
\end{equation}
For a packetized flow, it is shown in \cite{le_boudec_theory_2018} that this is equivalent to
%
\begin{align}\label{eq:arrival_minplus}
\sum_{k=m}^{n}l_k &\leq \alpha^+(A_n - A_m) \mfa m,n \in \N, m\leq n
\end{align}
where $\alpha^+$ is the right-limit of $\alpha$, and this is also equivalent to
\begin{align}
\label{eq:arrival_maxplus} A_n - A_m &\geq \alpha^{\downarrow}(\sum_{k=m}^{n}l_k)\mfa m,n \in \N, m\leq n.
\end{align}

 By \cite[Lemma~1.2.1]{le_boudec_network_2001}, if $\alpha$ is a bit-level arrival-curve for a flow, then so is its left limit, i.e., $\alpha^-$. Also, since input is packetized, $\alpha^+(0)$ is an upper bound on the size of any packet. Therefore, in the rest of this paper, we assume that $\alpha$ is left-continuous and \mbox{$\alpha^+(0)\geq L^{\max}$}. Frequently used bit-level arrival-curves are \cite[1.2.2]{le_boudec_network_2001}:
 \begin{itemize}
     \item the token-bucket (or leaky-bucket) arrival-curve with rate $r$ and burst $b$, defined by $\alpha(t)=rt+b; t>0$, $\alpha(0)=0$;
     \item the staircase function with burst $b$ and interval $\tau$ defined by $\alpha(t)=b\ceil{\frac{t}{\tau}}$. It expresses the constraint that 
the flow has at most $b$ bits within any interval of duration~$\tau$.
 \end{itemize}    

\subsection{g-regulation}\label{sec:greg}
The g-regulation constraint was introduced in \cite{changbook}; it specifies that the time inter-spacing between packets is lower bounded by a left-continuous function $g \in \Fincz$ \cite{changbook}\footnote{The original g-regularity in \cite{changbook} uses functions defined on $\N\cup\{0\}$ but it is simpler to consider functions defined on $\R^+$.}
A packetized input has g-regulation constraint if and only if for any packet indices $m,n$, where $m \leq n$:
\begin{equation}\label{eq:g_maxplus}
	A_n - A_m \geq g\left(\sum_{k=m}^{n-1} l_k\right).
\end{equation}
Note the difference between the bit-level arrival-curve in \eqref{eq:arrival_maxplus} and the g-regulation in \eqref{eq:g_maxplus}: for the g-regulation the length of the last packet, i.e., $l_n$, does not play a role \eqref{eq:g_maxplus} while it is included in \eqref{eq:arrival_minplus}. Since the argument of the function $g$ only takes values in a discrete set containing sums of packet sizes, we assume $g$ is left continuous at each point in this set.

\textit{Shifted-Rate Regulation} \cite[Section 6.2.1]{changbook} is a type of g-regulation with $g(x)~=~\frac{1}{r}[x-d]^+$, where $r$ is the regulation rate and $d$ is the regulation delay, i.e., the constraint is 
\begin{align}
	\label{eq:srr_maxplus}A_n - A_m \geq \frac{1}{r}\left[\left(\sum_{k=m}^{n-1} l_k\right)-d\right]^+, \forall m,n \in \N, m\leq n.
\end{align}

\textit{Length-Rate Quotient} (LRQ), introduced by \cite{specht_urgency_based_2016}, is a traffic regulation which specifies the minimum interspacing between two consecutive packets as a function of a regulation rate $r$ and the length of the earlier packet. i.e. the constraint is
\begin{equation}\label{eq:lrq}
	A_n - A_{n-1}\geq \frac{l_{n-1}}{r},~~\forall n \in \N, n \geq 2.
\end{equation}
By LRQ, the arrival of packet $n$ is constrained by the regulation rate, the arrival time and the length of the previous packet; the fact that there is dependency on only the last packet, renders the implementation of LRQ very simple. It is easy to see that, when $d=0$, \eqref{eq:srr_maxplus} is equivalent to \eqref{eq:lrq}, namely, LRQ is shifted-rate regulation with $d=0$; 
therefore, LRQ is a form of g-regulation with $g(x)=\frac{x}{r}$.

\subsection{Packet-level Arrival-Curve}\label{sec:pktarr}
The packet-level arrival-curve is mainly used in the context of IEEE TSN and IETF DetNet and expresses traffic constraints at the packet level. More formally, consider a left-continuous wide-sense increasing function \mbox{$N:\R^+ \rightarrow \N \cup \{0\}$}, where $N(t)$ is the cumulative number of packets observed on a flow of interest until time $t$ (excluded). Then, we say that the flow has a packet-level arrival-curve $\alphapcr{}:\R^+\rightarrow\N$, $\alphapcr{}(0)=0$,  if and only if
\begin{equation}\label{eq:pcr_floor}
N(t)-N(s) \leq \alphapcr{}(t-s), \mfa t\geq 0 \mand s\in[0,t].
\end{equation}
Similarly to the bit-level arrival-curve constraint, \eqref{eq:pcr_floor} is equivalent to: 
\begin{equation}\label{eq:pcr_minplus}
n-m+1 \leq \alphapcr{}^+(A_n-A_m) \mfa m,n \in \N, m\leq n.
\end{equation}
It is known that an arrival curve can always be assumed to be sub-additive (i.e. satisfy $\alphapcr{}(s+t)\leq \alphapcr{}(s)+\alphapcr{}(s)$), since otherwise it can be replaced by its sub-additive closure \cite{le_boudec_network_2001}. Also, following the same steps as the proof of \cite[Lemma~1.2.1]{le_boudec_network_2001}, we can prove that if $\alphapcr{}$ is a sub-additive packet-level arrival-curve for a flow, then so is its left limit $\alphapcr{}^-$. Also, since the input is packetized, $\alphapcr{}^+(0)$ gives an upper bound on one packet.  Therefore, in the rest of this paper, we assume that $\alphapcr{}$ is sub-additive, left continuous and  \mbox{$\alphapcr{}^+(0)\geq 1$}. 

\begin{figure}
	\centering
	\includegraphics[width=0.9 \linewidth]{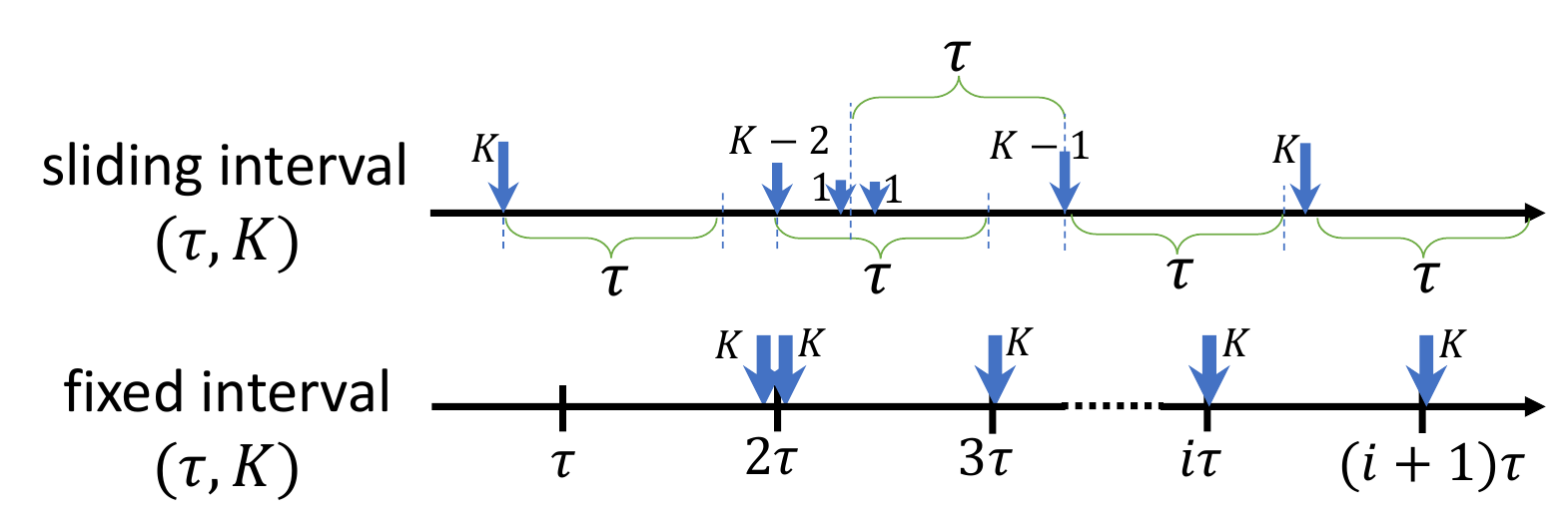}
	\caption{The two interpretations of TSN and DetNet interval. }
	\label{fig:tsn-traffic}
\end{figure} 

In IEEE TSN the traffic specifications for a flow is defined as \cite[Section 34.6.1]{ieee8021q}:
\begin{quote}
    ``... during each class measurement interval, it [a source] can place up to MaxIntervalFrames data frames, each no longer than MaxFrameSize into that stream’s queue.''
\end{quote}
Similarly, the traffic specification in IETF DetNet allows packet level constraint \cite[Section 5.5]{rfc9016} by the attributes ``Interval'' as the period of time in which the traffic specification is specified and ``MaxPacketsPerInterval'' as the maximum number of packets that a source transmits in one Interval.


By the above specifications, the interval (i.e. class measurement interval in TSN or Interval in IETF DetNet) can be interpreted differently as seen in Fig. \ref{fig:tsn-traffic}:

\begin{enumerate}
    \item Sliding interval $(\tau, K)$: the number of packets is limited by $K=$MaxIntervalFrames (or MaxPacketsPerInterval) at any sliding interval of duration $\tau$.
    \item Fixed interval $(\tau, K)$: the number of packets is limited by $K=$MaxIntervalFrames (or MaxPacketsPerInterval) at any reference interval of duration $\tau$; the reference intervals are consecutive and non-overlapping.
\end{enumerate}


\noindent The first interpretation, sliding interval $(\tau, K)$, is equivalent to
\begin{equation}\label{eq:inter1-2}
N(t+\tau)-N(t) \leq K;~~~\forall t\geq 0,
\end{equation}
which is also equivalent to saying that the flow has a packet-level arrival-curve given by
\begin{equation}\label{eq:cpcr}
    \alphapcr{}(t)=K\ceil{\frac{t}{\tau}}. 
\end{equation}
This is a staircase packet-level arrival-curve with burst equal to $K$ packets and period $\tau$.
It is easy to verify that it is sub-additive and left-continuous. Applying the right limit of \eqref{eq:cpcr} into \eqref{eq:pcr_minplus}, \eqref{eq:inter1-2} is equivalent to the following for any packet indices $m,n$, $m \leq n$:
\begin{align}\label{eq:cpcr_minplus_maxplus}
n-m+1 &\leq K \floor{\frac{A_n-A_m}{\tau}}+K.
\end{align}
The second interpretation, fixed interval $(\tau, K)$, is equivalent to saying that there exists some offset $\theta$ such that:
\begin{align}\label{eq:fixed-interval}
    \nonumber & N(\theta)=N(0)=0,\\
    N(\theta+(i+1)\tau)-&N(\theta + i\tau)\leq K ;~~~\forall i\in \N.
\end{align}
As shown in Lemma~\ref{lem:arr-fixed-interval} in Appendix~\ref{appendix:technical}, the second interpretation, fixed interval $(\tau, K)$, implies a packet-level arrival-curve with,
\begin{equation}\label{eq:fpcr}
\alphapcr{}(0)=0,~~\alphapcr{}(t)=K\ceil{\frac{t}{\tau}}+K : t>0,
\end{equation}
It is easy to verify that this function is sub-additive and left-continuous. However, the converse does not hold, i.e., it is not true that all flows that have the packet-level arrival curve in \eqref{eq:fpcr} satisfy fixed interval $(\tau, K)$ (because such an arrival curve allows $2K$ packets in an interval of duration less than $\tau$). Nonetheless, we show in Section~\ref{sec:delay} that the delay bound obtained using the packet-level arrival-curve in \eqref{eq:fpcr}, is tight for flows with the fixed interval $(\tau, K)$ regulation constraint.

Another form of packet-level regulation found in the literature is the token-bucket packet-level constraint \cite{le_boudec_theory_2018}, which limits the number of packets within any time interval $t$ to $\rho t+B$, where $\rho>0$ and $B\geq 1$ are respectively the packet rate and burst.
This constraint enjoys the superposition property, i.e., for an aggregation of flows each with such constraint, the superposition is constrained by a token-bucket packet-level constraint with $\rho$ and $B$ equal to the sum of the packet rates and bursts of the flows respectively. Such a constraint is equivalent to the following packet-level arrival-curve:
\begin{equation}\label{eq:lb-pkt}
\alphapcr{}(0)=0,~~\alphapcr{}(t)= \ceil{\rho t + B-1}: t>0.
\end{equation}
It can be easily verified that this function is sub-additive and left-continuous. Now, by \eqref{eq:pcr_minplus},  for any packet indices $m,n$, where $m \leq n$, \eqref{eq:lb-pkt} is equivalent to:
\begin{align}\label{eq:lpcr_minplus_maxplus}
n-m+1 &\leq \ceil{\rho(A_n - A_m) + B-1}.
\end{align}

\noindent The $(\lambda,\nu)$-constraint introduced in \cite{jiang_basic_2018}, is equal to the token-bucket packet-level arrival-curve with \mbox{$\rho=\lambda$} and \mbox{$B=\nu+1$}. The staircase packet-level arrival-curve in \eqref{eq:pcr_minplus} implies a token-bucket packet-level arrival-curve with $\rho=\frac{\tau}{K},~B=K$.


\subsection{Additional Definitions}\label{sec:notations}
For two functions $w,w'\in\Fincz$, the horizontal deviation from $w'$ to $w$ is defined as  \cite{le_boudec_network_2001} :
\begin{equation}\label{eq:horizontal}
    h(w',w) \eqdef \sup_{t \geq 0} \{w^{\downarrow}\left(w'(t)\right)-t\}.
\end{equation}
See Fig.~\ref{fig:horizontal}.

\begin{figure}[t]
	\centering
	\includegraphics[width=0.5 \linewidth]{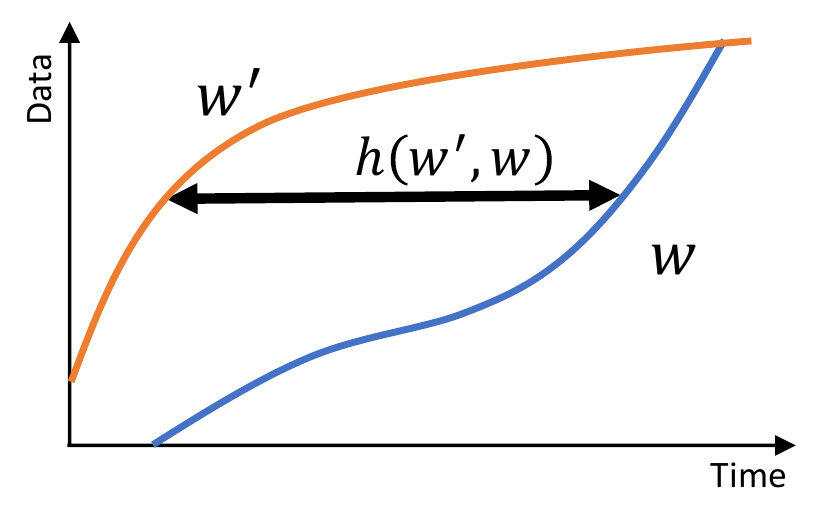}
	\caption{The horizontal deviation from $w'$ to $w$.}
	\label{fig:horizontal}
\end{figure}

A function $w\in\Fincz$ is called $c$-Lipschitz (for $c\geq 0$) if $\forall t_1,t_2\in\R^+$ \cite[Section 41.5]{eriksson2004applied}:
\begin{align}\label{eq:lips}
    |w(t_2)-w(t_1)| \leq c|t_2-t_1|.
\end{align}
A $c$-Lipschitz function is necessarily continuous and the slope of the function is within $[-c,c]$ at any point in time.
Fig.~\ref{fig:lips} shows examples of a $c$-lipschitz and non $c$-lipschitz functions.

\begin{figure}[t]
	\centering
	\includegraphics[width= \linewidth]{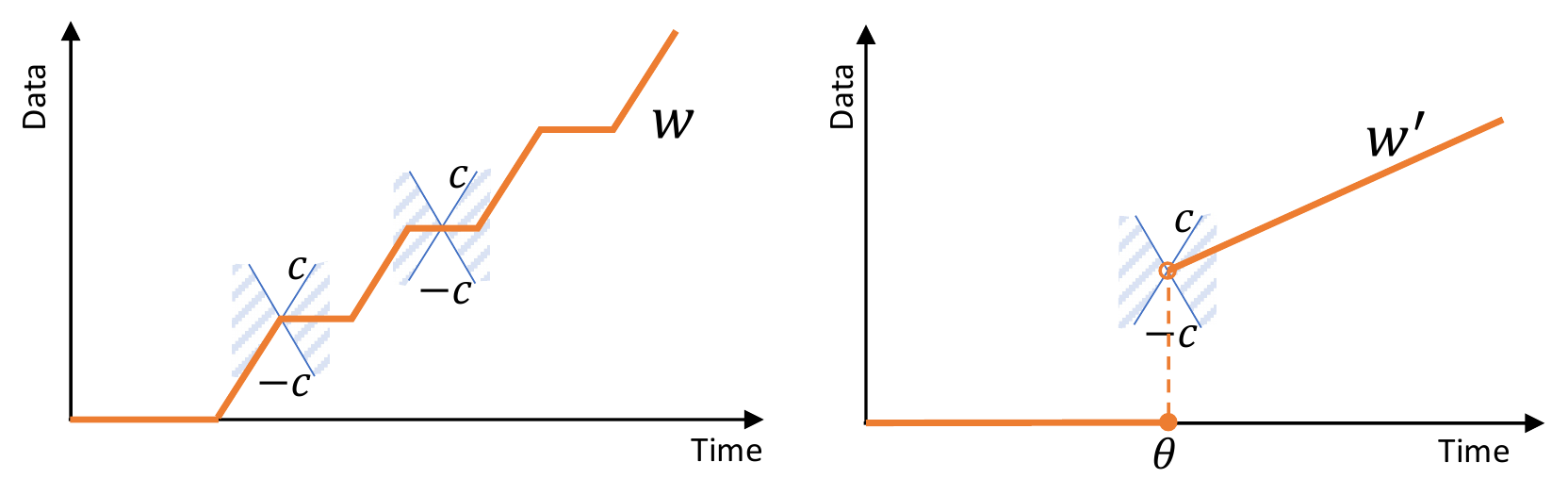}
	\caption{A $c$-lipschitz function $w$ (left, DRR service curve in \cite{tabatabaee2021deficit}) and a non $c$-lipschitz function $w'$ (FIFO residual service curve in \cite[Section 7.3.1]{bouillard_deterministic_2018}). The slope of the function is within $[-c,c]$ at any point in time for $w$. The function $w'$ is not continuous at time $\theta$.
	}
	\label{fig:lips}
\end{figure}

For $w\in\Fincz$, $w^+$ and $w^-$ are the right and left limits, defined as:
\begin{align}
    \forall x \in \R^+ ~:~ &w^+(x) = \lim_{\substack{\varepsilon\rightarrow 0\\\varepsilon>0}}w(x+\varepsilon),\\
    \forall x \in \R^+ ~:~ &w^-(x) = \lim_{\substack{\varepsilon\rightarrow 0\\\varepsilon>0}}w(x-\varepsilon).
\end{align}

\section{Relations among Flow Regulations}\label{sec:regModel}




We have seen various families of traffic regulations in Section~\ref{sec:sys}. We can find regulation-specific toolbox for delay analysis, e.g., \cite{le_boudec_network_2001,bouillard_deterministic_2018} for bit-level arrival curves. The immediate question is whether there is a relation among these regulation types. Such relations can open up new opportunities to apply the existing toolbox of one regulation type to another, that may lead to delay improvements, as we will see in Section~\ref{sec:delay}.
The goal of this section is to present such derivations between the regulation types. 

In the next proposition, we present a relation between g-regularity and bit-level arrival-curve. The derivation of bit-level arrival-curve from g-regularity already exists in literature \cite{changbook} and we put it here for completeness.
\begin{proposition}\label{pro:arr-g}
Consider a flow $f$. 
\begin{enumerate}
    \item If $f$ has a bit-level arrival-curve $\alpha(t)$, it also conforms to a $g$-regularity constraint with \mbox{$g(x)=\alpha^{\downarrow}(x+L^{\min}),~x>0;g(0)=0$}.
    \item If $f$ conforms to a $g$-regularity constraint, it also has a bit-level arrival-curve $\alpha(t) = g^{\downarrow}(t) + L^{\max}$.
    \item The sequential application of items 1 and 2 or items 2 and 1 gives a regulation constraint that is weaker than the initial one, except if all packets have the same size.
\end{enumerate}
\end{proposition}
The proof is in Appendix \ref{proof:arr-g}. Proposition~\ref{pro:arr-g} shows that even though bit-level arrival-curve and g-regularity are two different families of traffic constraints, they can be derived from each other in items (1) and (2). However, item (3) shows that when packets are of different sizes, such derivations lead to weaker constraints than the initial traffic constraints.

\begin{proposition}\label{pro:rate2reg}
	Consider a flow with packet-level arrival-curve $\alphapcr{}$, then:
	\begin{enumerate}
	\item The flow conforms to a $g$-regularity constraint with:
	\begin{equation}\label{eq:gregapp}
	g(x)= \alphapcr{}^{\downarrow}(\frac{x}{L^{\max}}+1),
	\end{equation}
	and to a bit-level arrival-curve constraint, $\alpha$, with:
	\begin{equation}\label{eq:acapp}
	\alpha (t) = L^{\max} \alphapcr{}(t).
	\end{equation}
	\item The $g$-regularity constraint in \eqref{eq:gregapp} is stronger than the bit-level arrival-curve constraint in \eqref{eq:acapp}; i.e., the sequential application of \eqref{eq:gregapp} and then item 1 of Proposition~\ref{pro:arr-g} gives \eqref{eq:acapp} while the sequential application of \eqref{eq:acapp} and then item 2 of Proposition~\ref{pro:arr-g} gives a worse constraint than \eqref{eq:gregapp}, except if all packets have the same size.
\end{enumerate}
\end{proposition}
The proof is in Appendix \ref{proof:rate2reg}. By item~(1), we can use the bit-level arrival-curve and g-regularity to compute delay bounds for flows with packet-level arrival-curves. In fact, for TSN/DetNet traffic regulation, such bit-level arrival-curve derivation already exists \cite{daigmorte_modelling} and is used for delay-bound computation. 
While in literature only the bit-level arrival-curve derivation is used to obtain delay bounds, item~(2) shows that the g-regularity derivation is a stronger constraint, and it leads to delay improvements as we will see in Theorem~\ref{thm:delay_pcr}. 


\section{Delay Bounds for Packet-Level Arrival-curves}\label{sec:delay}

The focus of this section is to find delay bounds for flows with packet-level arrival-curve that share a FIFO system described in Section~\ref{sec:sys}. To do so, we use a novel combination of $g$-regularity and bit-level arrival-curves derived from packet-level arrival-curves. 
We start with a technical lemma about queuing delay.


\begin{lemma}\label{lem:queuing-delay}
    Consider a FIFO system with service curve $\beta\in\Fincz$. Assume that we know some $w\in \Fincz$ such that for any two packet indices $m,n$, $m\leq n$, we have
    \begin{align}\label{eq:queuing-input}
        \sum_{k=m}^{n-1} l_k \leq w(A_n-A_m),
    \end{align}
    where $A_m,A_n$ are packet arrival times and $l_k$ is the length of packet $k$. Then, the queuing delay of the FIFO system is bounded by $h(w,\beta)$, i.e. for any packet index $n$:
    $$
    Q_n-A_n\leq h(w,\beta)
    $$
    where $h$ is the horizontal deviation.
\end{lemma}
The proof is Appendix~\ref{proof:queuing-delay}; it combines min-plus representation of the service curve and the max-plus representation of the input traffic in \eqref{eq:queuing-input}.

Using Lemma~\ref{lem:queuing-delay}, we obtain our first result on delay bounds, on which all delay bounds found in this paper are based. 


\begin{lemma}\label{lem:delay-alpha-g}
	Consider a FIFO system offering a service curve $\beta$ to the aggregate of the flows sharing it. Assume that (i) flow $1$ is the flow of interest and conforms to a $g$-regularity constraint, and (ii) flow $2$ represents the aggregate of the remaining flows sharing the FIFO system and has bit-level arrival-curve $\alpha$. In addition, assume that as soon as a packet starts to be transmitted, it is transmitted with rate $c$.
	Then, 
	
	(i) An upper bound on the response time of a packet with length $l$ of flow $1$ at this FIFO system is:
	\begin{equation}
	\Delta^{\A\G}(l)= h(g^{\uparrow} + \alpha^+,\beta) + \frac{l}{c},
	\end{equation}
	(ii) An upper bound on the response time of any packet of flow $1$ at this FIFO system is:
	\begin{equation}
	\Delta^{\A\G}= h(g^{\uparrow} + \alpha^+,\beta) + \frac{L_1^{\max}}{c}.
	\end{equation}
\end{lemma}
The proof is in Appendix~\ref{proof:delay-alpha-g}. 
We can now apply Lemma~\ref{lem:delay-alpha-g} to packet-level arrival curves and obtain the following theorem:


\begin{theorem}\label{thm:delay_pcr}
	Consider a FIFO system offering a service curve $\beta$ to the aggregate of the flows sharing it. Assume any flow $f\in\{1,\dots,M\}$ has $\alphapcr{f}$ as packet-level arrival-curve at the entrance of the FIFO system and has maximum packet length $L^{\max}_f$. In addition, assume that as soon as a packet starts to be transmitted, it is transmitted with rate $c$. Then, 

	(i) An upper bound on the response time of a packet with length $l$ of a flow $f$ at this FIFO system is:
	
	\begin{equation} \label{eq:pcr_bound1}
		\Delta^{\text{pkt}}(l)= h\left(\sum_{i=1}^{M}L^{\max}_i\alphapcr{i}^+- L_f^{\max},\beta\right) + \frac{l
		}{c}.
	\end{equation}
	
	(ii) An upper bound on the response time of any packet of flow $f$ at this FIFO system is:

    \begin{equation} \label{eq:pcr_bound2}
	\Delta^{\text{pkt}}= h\left(\sum_{i=1}^{M}L^{\max}_i\alphapcr{i}^+- L_f^{\max},\beta\right) + \frac{L^{\max}_f
	}{c}.
	\end{equation}
\end{theorem}
\begin{proof}
Use item 1 of Proposition~\ref{pro:rate2reg} to derive a g-regularity constraint for the flow of interest and a bit-level arrival-curve for the competing flows. Then apply Lemma~\ref{lem:delay-alpha-g}.
\end{proof}

The delay bound in Theorem~\ref{thm:delay_pcr} improves the state-of-the-art delay bounds at least when packet sizes are different. 
Let us observe the improvement for a simple case of rate-latency service-curves in the following example; the bounds are summarized in Table~\ref{table:bound-compare}.

\begin{application}
Consider a FIFO system with rate-latency service-curve $\beta(t)=R[t-T]^+$ connected to a link with transmission rate \mbox{$c> R$}. Assume flows $1$ and $2$ have packet-level arrival-curves $\alphapcr{}(t)=K\floor{\frac{t}{\tau}}$ and $\alphapcr{}'(t)=K'\floor{\frac{t}{\tau'}}$ respectively ($K,K'\in \N,~\tau,\tau'>0$). To avoid unbounded response time, we assume $\frac{KL_1^{\max}}{\tau}+\frac{K'L_2^{\max}}{\tau'}\leq R$. We want to compute a delay bound for flow $1$.

The classical approach is to first derive the bit-level arrive-curves corresponding to the two flows. By Proposition~\ref{pro:rate2reg}, they are:
\begin{align}
    \nonumber \alpha(t) &= L_1^{\max}\alphapcr{}(t)=KL_1^{\max}\floor{\frac{t}{\tau}},\\
    \alpha'(t) &= L_2^{\max}\alphapcr{}'(t)=K'L_2^{\max}\floor{\frac{t}{\tau'}}.
\end{align}
Then, a delay bound is computed by the classical network-calculus bound. For rate-latency service-curves, the classical network-calculus bound is improved by \cite{mohammadpour_improved_2019}, which is a special case of our result in Theorem~\ref{thm:response_fifo_arrival} (we will see in the next section).
Our approach, however, is to consider the original packet-level arrival-curves and directly apply Theorem~\ref{thm:delay_pcr}. Table~\ref{table:bound-compare} shows the obtained bounds for the above approaches.

Comparing the bound of Theorem~\ref{thm:delay_pcr} with the improved version of the classical approach, we have:
\begin{align}
    \Delta^\pkt - \Delta^\A = -(L_1^{\max} - L_1^{\min})\left(\frac{1}{R}-\frac{1}{c}\right) <0,
\end{align}
which implies that the bound obtained by Theorem~\ref{thm:delay_pcr} is strictly less than the improved classical approach, except when all packets of flow $1$ are of the same size ($L_1^{\min}=L_1^{\max}$).

\begin{table}[h]
\renewcommand{\arraystretch}{1.4} 
\caption{Comparison of the delay upper-bound for packet-level arrival-curve in the example of Section~\ref{sec:delay}.}
\label{table:bound-compare}
\begin{tabular}{|c|c|}
\hline
Approach                                                                                                                                                                          & Delay upper-bound                                                       \\ \hline
\begin{tabular}[c]{@{}c@{}}Classical approach: \\ bit-level arrival-curve \\ $+$ classical network-calculus bound\end{tabular}               & $\Delta^{\mathrm{NC}}=T + \frac{KL_1^{\max}+K'L_2^{\max}}{R}$                                \\ \hline
\begin{tabular}[c]{@{}c@{}} Improved approach: \\ bit-level arrival-curve  \\$+$ improved bound in Theorem~\ref{thm:response_fifo_arrival}\end{tabular} & $\Delta^\A=\Delta^{\mathrm{NC}} -L_1^\min \left(\frac{1}{R}-\frac{1}{c}\right)$   \\ \hline
\begin{tabular}[c]{@{}c@{}}Improved approach: \\ packet-level arrival-curve  \\$+$ new bound in Theorem~\ref{thm:delay_pcr}\end{tabular}                                 & $\Delta^\pkt=\Delta^{\mathrm{NC}}-L_1^{\max}\left(\frac{1}{R} - \frac{1}{c}\right)$ \\ \hline
\end{tabular}
\end{table}
\end{application}

The novel delay bound in Theorem~\ref{thm:delay_pcr} is derived from Lemma~\ref{lem:delay-alpha-g} and thus uses g-regularity constraint and bit-level arrival-curves derived from the packet-level arrival-curves. Therefore, it is legitimate to wonder whether it is the best possible bound derived from packet-level arrival curves. We show in the following theorem that this is indeed the case. 

\begin{theorem}	\label{thm:delay_pcr_tight}
    The bound of Theorem~\ref{thm:delay_pcr} is tight.
	
	More specifically, consider a $c$-Lipschitz service curve $\beta$, sub-additive left-continuous functions $\alphapcr{f}$ as packet-level arrival-curves for flow $f\in\{1,\dots,M\}$, with maximum packet lengths $L_f^{\max}$, and a transmission rate $c$. There exists a simulation trace of a FIFO system, shared between the $M$ flows, where a packet of flow $1$ reaches the bound in \eqref{eq:pcr_bound2} of Theorem \ref{thm:delay_pcr}.
\end{theorem}
The proof is in Appendix \ref{proof:delay_pcr_tight}; it consists in building a trajectory with \emph{greedy} sources, i.e. with sources for which the cumulative packet arrival function $N(t)$ satisfies \mbox{$N(t)=\alphapcr{} (t-t_0)$} for some offset $t_0$. 

Recall that there are two interpretations of the packet-level regulation constraint of TSN/DetNet (Section~\ref{sec:pktarr}). The former, sliding interval, is equivalent to a packet-level arrival curve constraint; the latter, fixed interval, implies a packet-level arrival curve constraint but is not equivalent to it. However, as we show in the next theorem, the delay bound obtained by using packet-level arrival-curves and Theorem~\ref{thm:delay_pcr} is the best possible bound for flows with TSN/DetNet constraints.
\begin{theorem}\label{thm:delay_tight_fixed_interval}
Consider a FIFO system offering a $c$-Lipschitz service curve $\beta$ to an aggregate of $M$ flows, where, as soon as a packet starts to be transmitted, it is transmitted with rate $c$. Assume that every flow $f\in\{1,\dots,M\}$ conforms to either the Sliding interval $(\tau_f, K_f)$ or the Fixed interval $(\tau_f, K_f)$ constraint of TSN/DetNet (Section~\ref{sec:pktarr}) and has maximum packet size $L_f^{\max}$. 
For every flow $f$, we can derive a packet-level arrival curve, by using \eqref{eq:cpcr}  or \eqref{eq:fpcr}, and apply Theorem~\ref{thm:delay_pcr} to obtain
a delay bound $\Delta^{\text{pkt}}_f$.

Then, for every flow $f$ and every $\varepsilon>0$, there exists a simulation trace of 
this system where a packet of flow $f$ experiences a delay in the interval $[	\Delta^{\text{pkt}}_f-\varepsilon, 	\Delta^{\text{pkt}}_f]$.
\end{theorem}
The proof is in Appendix \ref{proof:delay_tight_fixed_interval}. The main issue here is that a flow that conforms to Fixed interval $(\tau, K)$ has packet-level arrival curve given by \eqref{eq:fpcr}, but a greedy source for this arrival curve does not satisfy Fixed interval $(\tau, K)$ (as it generates $2K$ packets in one interval of duration $\tau$). We overcome this by considering, for every $\varepsilon>0$, a greedy source for the packet-level arrival curve $K\ceil{\frac{[t-\varepsilon]^+}{\tau}}+K$, which does conform to Fixed interval~$(\tau,K)$.

\section{Improvements on Existing Network-Calculus Delay-Bounds} \label{sec:ncimprov}
In this section, as a by-product of Lemmas~\ref{lem:queuing-delay} and \ref{lem:delay-alpha-g}, we derive novel delay bounds for flows with g-regularity constraints (Theorem~\ref{thm:greg-delay}) and bit-level arrival-curves (Theorem~\ref{thm:response_fifo_arrival}). 
The obtained delay bounds for flows with g-regularity constraint generalize the results of \cite{specht_urgency_based_2016}, which are specific to LRQ and strict-priority schedulers. Similarly, for flows with bit-level arrival-curve, our delay bounds improve on the ones obtained by classical network-calculus and generalize the results of  \cite{mohammadpour_improved_2019}, which are specific to rate-latency service-curves.


\begin{theorem}\label{thm:greg-delay}
	Consider a FIFO system offering a service curve $\beta$ to the aggregate of $n$ flows sharing it. Assume any flow $f\in\{1,\dots,M\}$ has $g_f$-regulation constraint at the entrance of the FIFO system and has maximum packet length $L^{\max}_f$. In addition, assume that as soon as a packet starts to be transmitted, it is transmitted with rate $c$. Then, 
	
	(i) An upper bound on the response time of a packet with length $l$ of a flow $f$ at this FIFO system is given as follows:
	\begin{equation}
		\Delta^G(l)= h\left(\sum_{i=1}^{M}g_i^{\uparrow} + \sum_{i=1}^{M}L_i^{\max}1_{i\neq f},\beta\right) + \frac{l}{c},
	\end{equation}
	
	(ii) An upper bound on the response time of any packet of flow $f$ at this FIFO system is given as follows:
	\begin{equation}
		\Delta^G= h\left(\sum_{i=1}^{M}g_i^{\uparrow} + \sum_{i=1}^{M}L_i^{\max}1_{i\neq f},\beta\right) + \frac{L_f^{\max}}{c}.
	\end{equation}
\end{theorem}
The proof consists in two steps; we first take the sum of the bit-level arrival-curves of the competing flows derived using item 1 of Proposition~\ref{pro:arr-g}; second, we apply Lemma~\ref{lem:delay-alpha-g}.

For the most common form of g-regularity constraint, i.e., LRQ, the current available delay bound is by \cite{specht_urgency_based_2016} which only applies to strict-priority schedulers; under the same assumption, Theorem~\ref{thm:greg-delay} gives the same delay bounds as \cite{specht_urgency_based_2016}. The delay bounds presented in \cite{specht_urgency_based_2016} do not apply to a vast majority of schedulers, e.g., WRR, Interleaved WRR, and DRR. Theorem~\ref{thm:greg-delay}, however, can be used to obtain delay bounds for such scheduling mechanisms. We provide a numerical illustration of Theorem~\ref{thm:greg-delay} for a DRR scheduler in Section~\ref{sec:eval-thm2}.

Next, we present the delay bounds for flows with bit-level arrival-curves. It follows from the definition that the bit-level arrival-curves can be superimposed: an aggregate of flows conforms to a bit-level arrival-curve equal to the sum of the bit-level arrival-curves of the flows. Hence, in the next theorem, we assume an aggregate of flows (that share a FIFO system with the flow of interest) is constrained by one bit-level arrival-curve constraint.

\begin{theorem}\label{thm:response_fifo_arrival}
	Consider a FIFO system offering a service curve $\beta$ to the aggregate of the flows sharing it and such that, as soon as a packet starts to be transmitted, it is transmitted with rate $c$. Flow $1$ is the flow of interest and flow $2$ represents the aggregate of the remaining flows sharing the FIFO system. Assume that  that flows $1,2$ have bit-level arrival-curves $\alpha,~\alpha'$ respectively; then, 
	
	(i) an upper bound on the response time of a packet with length $l$ of flow $1$ at this FIFO system is
	\begin{equation}\label{eq:delay-arr-l}
		\Delta^A(l) = h(\alpha^+ + {\alpha'}^+-l,\beta) + \frac{l}{c};
	\end{equation}
	
	(ii) if $\beta$ is $c$-Lipschitz, an upper bound on the response time of any packet of flow $1$ at this FIFO system is
	\begin{equation}\label{eq:delay-arr-flow}
		\Delta^A = h(\alpha + {\alpha'}-L^{\min}_1,\beta) + \frac{L^{\min}_1}{c}.
	\end{equation}	
\end{theorem}
The proof is in Appendix \ref{proof:response_fifo_arrival}. It consists in two steps. We first derive the queuing delay bound for the FIFO system using Lemma~\ref{lem:queuing-delay} and then we add the transmission time for the packet of interests. To obtain the per-flow delay-bound in item (ii), we take the supremum of per-packet delay-bound in item (i) for all the range of packet lengths, i.e., $[L_1^\min,L_1^\max]$. When the service-curve is $c$-lipschitz (formally defined in \eqref{eq:lips}), the supremum is achieved at \mbox{$l=L_1^\min$} (in Appendix~\ref{appendix-example}, we show that this does not always hold for a non $c$-Lipschitz service-curve). The $c$-lipschitz condition in Theorem~\ref{thm:response_fifo_arrival} implies that the slope of the service curve is not more than the transmission rate $c$ at any point in time. In fact, the rate-latency service-curves as well as the existing service-curves for the common scheduling mechanisms in time-sensitive networks, e.g., TSN schedulers with credit-based shapers \cite{mohammadpour_latency_2018,zhao_timing_analysis_2018}, DRR \cite{tabatabaee2021deficit,boyer2012drr}, Interleaved WRR\cite{tabatabaee2021interleaved}, are $c$-lipschitz; therefore, the bound in Theorem~\ref{thm:response_fifo_arrival} can be applied to such systems.

Theorem~\ref{thm:response_fifo_arrival} gives the same delay bound as \cite{mohammadpour_improved_2019} for rate-latency service-curves. For a number of scheduling mechanisms, e.g., DRR and Interleaved WRR, the rate-latency service-curves are not accurate; to this end, better non rate-latency service-curves are obtained that provide more precise characterization for these mechanisms \cite{tabatabaee2021deficit,tabatabaee2021interleaved}. For such non rate-latency service-curves, the delay-bound improvement of \cite{mohammadpour_improved_2019} cannot be used; instead, the existing approach is to use classical network calculus to obtain a delay bound, which is \mbox{$h(\alpha+\alpha',\beta)$}. This bound is improved by Theorem~\ref{thm:response_fifo_arrival}. We provide a numerical illustration of Theorem~\ref{thm:response_fifo_arrival} for a non rate-latency service-curve in Section~\ref{sec:eval-thm1}.



\section{Comparison of the Different Bounds}\label{sec:comparison}
In sections \ref{sec:delay} and \ref{sec:ncimprov} we obtained delay bounds for flows with various regulation types. Moreover, we showed in Section~\ref{sec:regModel} that regulation types can be derived from each other. Combining the mentioned results, we can derive a regulation type from another and then apply the corresponding theorem, e.g., deriving bit-level arrival-curve from g-regularity using item (1) of Proposition~\ref{pro:arr-g} and then apply Theorem~\ref{thm:response_fifo_arrival}. Now, the question is how such delay bounds are compared with the ones obtained by directly applying the theorem to initial flow constraint. The goal of this section is address this question.

We start with packet-level arrival-curve. Our tight delay bound for packet-level arrival-curves in Theorem~\ref{thm:delay_pcr} uses the g-regularity derivation of Proposition~\ref{pro:rate2reg} for the flow of interests. The other approach is to use bit-level derivation of packet-level arrival-curves in Proposition~\ref{pro:rate2reg}; this is indeed the state-of-the-art approach. We already showed in Table~\ref{table:bound-compare} that this leads to sub-optimal delay bounds for rate-latency service-curves. The following proposition shows that this is not accidental and is true in general. 

\begin{proposition}\label{pro:comp-pck}
Consider the assumptions in Theorem~\ref{thm:delay_pcr}. Using Proposition~\ref{pro:rate2reg} we can derive a bit-level arrival-curve and, by applying Theorem~\ref{thm:response_fifo_arrival},  obtain a delay bound $\Delta^\A_f$ for every flow $f$. Let $\Delta^{\pkt}_f$ be the delay bound obtained by a direct application of Theorem~\ref{thm:delay_pcr}. Then $\Delta^{\pkt}_f\leq \Delta^\A_f$. Furthermore, if packets of flow $f$ have a constant size ($L_f^\min=L_f^\max$) then $\Delta^{\pkt}_f= \Delta^\A_f$, else $\Delta^{\pkt}_f< \Delta^\A_f$, in general.
\end{proposition}
The proof is in Appendix~\ref{proof:comp-pck}. In the above, ``in general" means that we can find schedulers for which the inequality is strict, for example when the service curve is rate-latency with rate $R<c$.

Similarly, we evaluate the delay bounds achieved by deriving g-regularity from bit-level arrival-curve [resp. bit-level arrival-curve from g-regularity] and then applying Theorem~\ref{thm:greg-delay} [resp. Theorem~\ref{thm:response_fifo_arrival}].
Then, the question is whether we obtain the same delay bounds by applying theorems \ref{thm:greg-delay} and \ref{thm:response_fifo_arrival} respectively to the derived g-regularity constraint and bit-level arrival-curve. The following proposition shows that the answer to this question is negative and the obtained delay bound is generally worse, except for flows with packets of constant size.
\begin{proposition}\label{pro:compare-g-arr-delay}
Consider the FIFO system assume in Theorem~\ref{thm:response_fifo_arrival}.
	\begin{enumerate}
		\item Assume that the flows have bit-level arrival-curves. Then, Theorem \ref{thm:response_fifo_arrival} gives a better response time upper bound than the sequential application of item 1 of Proposition \ref{pro:arr-g} and Theorem~\ref{thm:greg-delay}. Specifically: \mbox{$\Delta^A_f \leq \Delta^{G}_f$}; furthermore, if packets of flow $f$ have a constant size ($L_f^\min=L_f^\max$) then $\Delta^A_f = \Delta^{G}_f$ and else \mbox{$\Delta^A_f < \Delta^{G}_f$} in general.  
		\item Assume that the flows have $g$-regularity constraints. Then, Theorem \ref{thm:greg-delay} gives a better response time upper bound than the sequential application of item 2 of Proposition \ref{pro:arr-g} and Theorem~\ref{thm:response_fifo_arrival}. Specifically: \mbox{$\Delta^A_f \leq \Delta^{G}_f$}; furthermore, if packets of flow $f$ have a constant size ($L_f^\min=L_f^\max$) then $\Delta^A_f = \Delta^{G}_f$ and else \mbox{$\Delta^G_f < \Delta^{A}_f$} in general.    
	\end{enumerate}
\end{proposition}
The proof is in Appendix~\ref{proof:compare-g-arr-delay}. 

Propositions \ref{pro:comp-pck} and \ref{pro:compare-g-arr-delay} indicate that the best delay bound for a given flow is obtained by directly applying the theorem corresponding to its initial constraint, as delay bounds obtained from derived flow constraints are pessimistic.

\section{Numerical Illustration}\label{sec:eval}
This section provides three example applications of the theorems presented in this paper to commonly used schedulers. To compute the delay bounds, we use the RealTime-at-Work (RTaW) tool\cite{RTaW-Minplus-Console} that has efficient implementation of network calculus operations.

\subsection{Flows with packet-level arrival-curve}\label{sec:eval-thm3}
Consider a FIFO system with TSN scheduler and Credit-based Shapers (CBSs) with per-class FIFO queuing \cite{mohammadpour_latency_2018}; from highest to lowest priority, the classes are CDT, A, B, and Best Effort (BE). The CBSs are used separately for classes A and B. The CBS parameters $\mathit{idleslope}$s are set to $50\%$ and $25\%$ of the link rate, $c=1$~Gbps, respectively for classes A and B. The CDT traffic has a token-bucket arrival curve with rate $6.4$~Kbps and burst $64$~Bytes. The maximum packet length of class BE is $1.5$~KB. Using the results in \cite{mohammadpour_latency_2018}, a rate-latency service curve offered to class~A has latency $T_A=12.5\mu$s and rate $R_A=499.92~\mbps$, and the one offered to class~B has $T_B=36.6\mu$s and rate $R_B=~249.75\mbps$.
\begin{table}[h] 
\caption{Flow information for Section \ref{sec:eval-thm3}.}
\begin{tabular}{|ccc||ccc|}
\hline
\multicolumn{3}{|c|}{Class A} & \multicolumn{3}{c|}{Class B} \\ \hline
\multicolumn{1}{|c|}{id} & \multicolumn{1}{c|}{$L^{\max}$ (Bytes)} & period (ms) & \multicolumn{1}{c|}{id} & \multicolumn{1}{c|}{$L^{\max}$ (Bytes)} & period (ms) \\ \hline
\multicolumn{1}{|c|}{1} & \multicolumn{1}{c|}{1442} & 16 & \multicolumn{1}{c|}{6} & \multicolumn{1}{c|}{1438} & 64 \\ \hline
\multicolumn{1}{|c|}{2} & \multicolumn{1}{c|}{185} & 4 & \multicolumn{1}{c|}{7} & \multicolumn{1}{c|}{619} & 64 \\ \hline
\multicolumn{1}{|c|}{3} & \multicolumn{1}{c|}{537} & 16 & \multicolumn{1}{c|}{8} & \multicolumn{1}{c|}{773} & 128 \\ \hline
\multicolumn{1}{|c|}{4} & \multicolumn{1}{c|}{414} & 4 & \multicolumn{1}{c|}{9} & \multicolumn{1}{c|}{459} & 128 \\ \hline
\multicolumn{1}{|c|}{5} & \multicolumn{1}{c|}{350} & 8 & \multicolumn{1}{c|}{10} & \multicolumn{1}{c|}{592} & 128 \\ \hline
\end{tabular}
\label{tbl:flows}
\end{table}
There are $5$ periodic flows for each of classes A and B; each flow send $1$ packet at each period. Table \ref{tbl:flows} shows the flow information. We want to compute delay bounds for flows $1$ and $6$ in classes A and B.

The state-of-the-art approach is to obtain bit-level arrival-curve for all flows using Proposition \ref{pro:rate2reg}; then by Theorem~\ref{thm:response_fifo_arrival} we obtain a delay bound of $63.58\us$ for flow $1$ (class A) and $158.47\us$ for flow $6$ (class B). However, we can also directly apply Theorem~\ref{thm:delay_pcr} to obtain delay bounds; for flows $1$ and $6$ we obtain $50.46\us$ and $126.32\us$ respectively, which shows  $20$\% improvement in both cases.

\subsection{Flows with g-regulation}\label{sec:eval-thm2}
Consider a FIFO system with aggregate queuing and DRR arbitration policy, with $n=8$ queues sharing a link with rate $c=1$~Gbps. Assume all flows have maximum packet length of $L^{\max}=1.5$~KB and the queues have same quantum value $Q=L^{\max}$. Then, similarly to the previous case, we obtain a service curve offered to any aggregate queue. Now, assume a flow of interest, conforms to LRQ with rate $r$, shares a queue with a number of other flows with LRQ regulation where the sum of their maximum packet lengths is $5$~KB. Also assume that  the minimum packet length of the flow is $100$~Bytes.

We obtain a bit-level arrival-curve constraint for the flows using Proposition~\ref{pro:arr-g}. Then by Theorem \ref{thm:response_fifo_arrival} we obtain a delay bound of $632.75\us$ for the flow of interest. Using the results of this paper, we can also directly apply Theorem \ref{thm:greg-delay} to obtain delay bounds that gives $552.74\us$, which shows  $12.6$\% improvement.

\subsection{Flows with bit-level arrival-curve}\label{sec:eval-thm1}
Consider a FIFO system with per-flow queuing and DRR arbitration policy, with $n=16$ queues sharing a link with rate $c=1$~Gbps. Assume all flows have maximum packet length of $L^{\max}=1.5$~KB and the queues have same quantum value $Q=L^{\max}$. Then, by \cite[Theorem 1]{tabatabaee2021deficit}, we obtain a service curve offered to any queue, 
\begin{align}
    \nonumber \beta(t) &= \left(\theta \circ \zeta\right) (t)+\min\left(\left[ct-2(n-1)(2L^{\max}-\epsilon)\right]^+,\epsilon\right),
\end{align}
with,
\begin{align}
    \nonumber\theta(t) &= \inf_{0\leq s\leq t}\left\{t-s+Q\floor{\frac{s}{(n-1)Q}}\right\}\\
    \nonumber \zeta(t) &= [ct-(n-1)(4Q-\epsilon)+\epsilon]^+,
\end{align}
and we set $\epsilon=1$~Byte.

Consider a flow with $\alpha(t)=rt+2L^{\max}$ and minimum packet length of $0.5$~KB. Since, the service curve is not rate-latency, we cannot apply the improvement in \cite{mohammadpour_improved_2019}. Hence, the state-of-the-art approach is to use the classical network-calculus bound, $h(\alpha,\beta)$; it is equal to $915.88\us$. Using our improved delay bound in Theorem \ref{thm:response_fifo_arrival}, the delay bound is reduced to $743.88\us$, which improves the classical network-calculus bound by $18.8$\%.

\section{Conclusion}\label{sec:conclusion}
We presented a theory to compute delay bounds for flows with packet-level 
arrival-curve, which improves the state-of-the-art bounds. The improvement is made possible by a novel modelling of packet-level arrival-curve with g-regularity and bit-level arrival-curve together with exploiting the information on the transmission rate. Our method of proof led to delay improvement for flows with g-regularity constraint and bit-level arrival-curve.
In time-sensitive networks, this result can open a discussion on the operation of flow re-shaping \footnote{Flow re-shaping refers to the process of recreating the arrival curve of a flow as its source}: even though in such networks the traffic specification at a source is at the packet level, flow re-shaping is performed at the bit level, based on a bit-level arrival-curve derived from the flow constraints at the source. As a result of the analysis of this paper, using packet-level traffic re-shaping leads to better delay bounds at the intermediate routers and switches. As the operation of bit-level re-shaping mechanisms is mainly based on full reception of a packet, e.g., as in IEEE802.1 Qcr Asynchronous Traffic Shaping, an implementation of packet-level re-shaping appear to be feasible and even simpler. 



\section{Acknowledgements}
This work was supported by Huawei Technologies Co., Ltd. in the framework of the project Large Scale Deterministic Network.

\bibliographystyle{IEEEtran}
\bibliography{ref}

\appendices
\twocolumn[

\begin{@twocolumnfalse}
 \begin{center}   
{\Large\textbf{Supplementary Material}}\\
{\textbf{Improved Network Calculus Delay Bounds in Time-Sensitive Networks}\\
\textit{Ehsan Mohammadpour, Eleni Stai, Jean-Yves Le Boudec}}
\end{center}
  \end{@twocolumnfalse}
  ]
  
\section{Technical Prerequisites}\label{appendix:technical}
Consider a function $f\in \Finc$, then by \cite[Section 10.1]{liebeherr_duality_2017}\cite{le_boudec_theory_2018},
\begin{align}
    \label{eq:jylb_theory_1}(f^+)^{\downarrow} &= f^{\downarrow},\\ 
    \label{eq:jylb_theory_2}(f^{\downarrow})^+ &= f^{\uparrow},\\
    \label{eq:jylb_theory_3}(f^{\uparrow})^- &= f^{\downarrow},\\
    \label{eq:jylb_theory_4}(f^-)^{\uparrow} &= f^{\uparrow},\\
    \label{eq:liebeherr_P6} \forall y\in \R^+~:~f(x)\leq y &\implies x \leq f^{\uparrow}(y), \\
    \label{eq:liebeherr_P8} \forall y\in \R^+~:~f(x) \geq y &\implies x \geq f^{\downarrow}(y),
\end{align}
where $f^+$ is the right limit of the function $f$. Furthermore, by \cite[Section 10.1]{liebeherr_duality_2017}\cite{boyer_common}, if $f$ is right continuous:
\begin{align}
	\label{eq:liebeherr_lemma10.1d} f&=(f^{\downarrow})^{\uparrow},\\
	\label{eq:comp-down} \forall w \in \Finc~:~(f \circ w )^{\downarrow}(x)&=(w^{\downarrow}\circ f^{\downarrow})(x), 
\end{align}
and if $f$ is left continuous:
\begin{align}
	\label{eq:liebeherr_lemma10.1c} f&= (f^{\uparrow})^{\downarrow},\\ 
	\label{eq:comp-up}\forall w \in \Finc~:~(f \circ w )^{\uparrow}(x)&=(w^{\uparrow}\circ f^{\uparrow})(x).
\end{align}
Note that $\circ$ is the composition operator, i.e., $(f \circ w )(x)=f(w(x))$.
\begin{lemma}\label{lem:upper-lower-conseq}
Consider a left-continuous function $f\in\Finc$. Then
    $\left(\left(f^\downarrow\right)^\uparrow\right)^- = f$.
\end{lemma}
\begin{proof}
    By \eqref{eq:jylb_theory_3}, we have:
    \begin{align}
        \left(\left(f^\downarrow\right)^\uparrow\right)^- = \left(\left(\left(f^\uparrow\right)^-\right)^\uparrow\right)^-.
    \end{align}
    Then, by \eqref{eq:jylb_theory_4}:
    \begin{align}
        \left(\left(\left(f^\uparrow\right)^-\right)^\uparrow\right)^- =\left(\left(f^\uparrow\right)^\uparrow\right)^-.
    \end{align}
    Then, by \eqref{eq:jylb_theory_3}:
    \begin{align}
        \left(\left(f^\uparrow\right)^\uparrow\right)^- =\left(f^\uparrow\right)^\downarrow = f,
    \end{align}
    where the last equality is by \eqref{eq:liebeherr_lemma10.1c} and left-continuity of $f$.
\end{proof}
\begin{lemma}\label{lem:pseudo_inverse_Lipschitz}
	Let $f$ be a wide-sense increasing and $c$-Lipschitz function. Then, for $x'\geq x$, we have:
	\begin{equation}
	f^{\downarrow}\left(x'\right) -f^{\downarrow}\left(x\right) \geq \frac{x'-x}{c}.
	\end{equation}
\end{lemma}
\begin{proof}
	According to the definition of Lipschitz continuity and because $f$ is wide-sense increasing, we have for $t'\geq t$:
    \begin{equation}
    f\left(t'\right) -f\left(t\right) \geq c(t'-t).
    \end{equation}
    Assume that 	$f^{\downarrow}\left(x'\right)=t'$ and $f^{\downarrow}\left(x\right)=t$. Due to Lipschitz continuity, $f$ is continuous. 
    Therefore, from \eqref{eq:lsi-def}, $t=f^{\downarrow}\left(x\right)=\inf \{s \geq 0| f(s) \geq x \}=\sup \{s \geq 0| f(s) < x \}=\sup \{s \geq 0| f(s) \leq  x \}$. Thus, $f\left(t\right)=x$. Similarly, we can show that $f\left(t'\right)=x'$. Therefore, we obtain
    	\begin{align}
    f^{\downarrow}\left(x'\right) -f^{\downarrow}\left(x\right) = t'-t \\
    \geq  \frac{f\left(t'\right) -f\left(t\right)}{c}= \frac{x'-x}{c}, 
    \end{align}
which completes the proof. 
\end{proof}

\begin{lemma}\label{lem:arr-fixed-interval}
    A flow with fixed interval ($\tau,K$) conforms to a packet-level arrival-curve $\alphapcr{}$ with
    \begin{align}
        \alphapcr{}(0)=0,~~\alphapcr{}(t)=K\ceil{\frac{t}{\tau}}+K : t>0.
    \end{align}
\end{lemma}
\begin{proof}
First, since \mbox{$N(s)-N(s) = 0; \forall s\geq 0$}, by \eqref{eq:pcr_floor}  \mbox{$\alphapcr{}(0)=0$}.

Next we prove the statement for $t>0$. For all $t>0$, there exists some $0\leq t'<\tau$ such that
    \begin{align}\label{eq:fixed-interval-proof-t}
        t = i \tau + t',~~~i=\ceil{\frac{t}{\tau}}.
    \end{align}
Now, consider a time instant $s$. We cover the two cases $s\leq \theta$ and $s>\theta$ separately.
\begin{itemize}
    \item $0\leq s\leq \theta$. Then by \eqref{eq:fixed-interval} $N(s)=0$. Therefore:
    \begin{align}\label{eq:fixed-interval-proof-1}
        N(s+t)-N(s) = N(s+t) \leq N(\theta +t).
    \end{align}
    By \eqref{eq:fixed-interval-proof-t}:
    \begin{align}
        \nonumber N(\theta +t)&= N(\theta +i \tau + t')\\
        \nonumber &= \left[N(\theta +i \tau + t') - N(\theta +i \tau)\right]\\
        \nonumber &+\left[N(\theta +i \tau) - N(\theta +(i-1) \tau)\right]+\dots\\
        &+\left[N(\theta + \tau) - N(\theta)\right] + N(\theta)
    \end{align}
    Using the definition of fixed interval constraint in \eqref{eq:fixed-interval}, we have:
    \begin{align}
        N(\theta +t) \leq (i+1) K +N(\theta)=(i+1) K
    \end{align}
    Since $i=\ceil{\frac{t}{\tau}}$, by \eqref{eq:fixed-interval-proof-1} we have:
    \begin{align}\label{eq:fixed-interval-proof-3}
        N(s+t)-N(s) \leq K \ceil{\frac{t}{\tau}} + K.
    \end{align}
    \item $s>\theta$. Then, there exists some $n\in \N$ such that:
    \begin{align}
        \theta + n\tau \leq s <\theta+(n+1)\tau.
    \end{align}
    Therefore,
    \begin{align}
        N(s+t)-N(s) \leq N(\theta+(n+1)\tau+t)-N(\theta + n\tau).
    \end{align}
    By \eqref{eq:fixed-interval-proof-t}, the above equation gives:
    \begin{align}\label{eq:fixed-interval-proof-2}
        \nonumber N(s+t)&-N(s) \leq N(\theta+(n+1)\tau+i \tau + t')\\
        \nonumber &-N(\theta + n\tau)=\big[N(\theta+(n+i+1)\tau + t')\\
        \nonumber&-N(\theta+(n+i+1)\tau)\big]+\\
        \nonumber&\big[N(\theta+(n+i+1)\tau)-N(\theta+(n+i)\tau)\big]\\
        &+\dots+\big[N(\theta+(n+1)\tau)-N(\theta+n \tau)\big].
    \end{align}
    Since $i=\ceil{\frac{t}{\tau}}$, \eqref{eq:fixed-interval-proof-2} gives:
    \begin{align}\label{eq:fixed-interval-proof-4}
        N(s+t)&-N(s) \leq (i+1) K = K \ceil{\frac{t}{\tau}} + K.
    \end{align}
\end{itemize}
By \eqref{eq:fixed-interval-proof-3} and \eqref{eq:fixed-interval-proof-3}, for all $s\geq 0$:
\begin{align}
    N(s+t)&-N(s) \leq K \ceil{\frac{t}{\tau}} + K = \alphapcr{}(t),
\end{align}
which proves the lemma.
\end{proof}


\section{Proofs of Theorems and Propositions}

\subsection{Proof of Proposition \ref{pro:arr-g}}\label{proof:arr-g}
1) Since the flow has an arrival curve $\alpha$, by \eqref{eq:arrival_minplus} for any packet index $m,n$ such that $m\leq n$, we have:
\begin{equation}
    \sum_{i=m}^{n} l_i \leq \alpha^+(A_n - A_m).
\end{equation}
By excluding the last packet from the left hand-side of the inequality, we have:
\begin{equation}\label{eq:arr2reg_1}
    \sum_{i=m}^{n-1} l_i \leq \alpha^+(A_n - A_m) - l_n \leq \alpha^+(A_n - A_m) - L^{\min}.
\end{equation}
We define $h(t) =[t - L^{\min}]^+$ and $x = \sum_{i=m}^{n-1} l_i$. Then, \eqref{eq:arr2reg_1} can be rewritten as:
\begin{equation}\label{eq:arr2reg__2}
	x  \leq (h\circ \alpha^+)(A_n - A_m).
\end{equation}
Applying \eqref{eq:liebeherr_P8} to \eqref{eq:arr2reg__2}:
\begin{equation}\label{eq:arr2reg__3}
	A_n - A_m \geq (h\circ \alpha^+)^{\downarrow} \left(x\right).
\end{equation}
Since $h$ is continuous, by \eqref{eq:comp-down} we have:
\begin{equation}
	A_n - A_m \geq ((\alpha^+)^\downarrow \circ h^{\downarrow}) \left(x\right).
\end{equation}
Applying \eqref{eq:jylb_theory_1}, we have:
\begin{equation}
	A_n - A_m \geq (\alpha^\downarrow \circ h^{\downarrow}) \left(x\right).
\end{equation}
Finally, by \eqref{eq:lsi-def}, we have $h^{\downarrow}(x)=x+L^{\min}, x>0;h^{\downarrow}(0)=0$. Then, for $x>0$:
\begin{equation}
	A_n - A_m \geq \alpha^\downarrow(x+L^{\min})=g(x),
\end{equation}
and for $x=0$:
\begin{equation}
	A_n - A_m \geq \alpha^\downarrow(0)=0=g(0).
\end{equation}

2) By [\cite{changbook}, Lemma 6.2.8], $g^{\uparrow}(t) + L^{\max}$ is an arrival curve for the flow. Since, the left limit of an arrival curve is also an arrival curve for the flow, we have:
\begin{align}
    \nonumber\alpha(t) &= \lim_{\substack{\varepsilon\rightarrow 0\\\varepsilon>0}}g^{\uparrow}(t-\varepsilon) + L^{\max} = \left(g^{\uparrow}\right)^-(t)+L^{\max}\\
    &= g^{\downarrow}(t)+L^{\max}.
\end{align}
The last equality is by \eqref{eq:jylb_theory_3}.

3) Assume a flow with arrival curve $\alpha$ and let us first apply item 1 and then item 2. By item 1, we can find a $g$-regularity constraint:
		\begin{equation}\label{eq:gregobt}
			g(x) = \alpha^{\downarrow}(x+L^{\min})=(\alpha^\downarrow \circ h)(x),
		\end{equation}
	where $h(x)=x+L^{\min}$. Now, we apply item 2 to the obtained $g$-regularity constraint in \eqref{eq:gregobt} and derive an arrival curve $\alpha'$:
		\begin{align}
			\nonumber \alpha'(t)&=g^{\downarrow}(t) + L^{\text{max}}  = (\alpha^\downarrow \circ h)^\downarrow (t)+ L^{\text{max}}\\
			&=\left((\alpha^\downarrow \circ h)^\uparrow\right)^- (t)+ L^{\text{max}}, 
		\end{align}
		The last equality is by \eqref{eq:jylb_theory_3}.
		Due to left continuity of $\alpha^\downarrow$, by \eqref{eq:comp-up} we have:
	\begin{align}
		\nonumber \alpha'(t)&=\left(\left(h^\uparrow \circ \left(\alpha^\downarrow\right)^\uparrow\right)\right)^- (t)+ L^{\text{max}}\\
		&=\left[\left(\left(\alpha^\downarrow\right)^\uparrow\right)^-(t) - L^{\text{min}}\right]^+ + L^{\text{max}}. 
	\end{align}
		Note that $h^\uparrow(t) = [t - L^{\min}]^+=\max(t - L^{\min},0)$. Since  $\alpha$ is left continuous and $\alpha(t)\geq L^{\max}$, by Lemma~\ref{lem:upper-lower-conseq} we have:
		\begin{align}\label{eq:prop1_2_remark__1}
			\alpha'(t)=\alpha(t)- L^{\min} + L^{\max}.
		\end{align}
		Eq. \eqref{eq:prop1_2_remark__1} shows that by applying item 1 and then item 2, the obtained arrival curve is not the same as the initial one, i.e., $\alpha\neq \alpha'$, except from the case that $L^{\max} = L^{\min}$ (when all packets have the same length).
	
		 Let us now examine the opposite direction. Assume a flow has $g$-regularity constraint. By applying item 2, we can find an arrival curve, $\alpha$:
		\begin{equation}\label{eq:der_arrival}
			\alpha(t) = g^{\downarrow}(t) + L^{\max} =(h \circ g^{\downarrow})(t),
		\end{equation}
	where $h(t)=t+L^{\max}$. 
		Now, we apply item 1 to the obtained arrival curve \eqref{eq:der_arrival} and derive a $g'$-regularity constraint:
		\begin{align}\label{eq:prop2_1_remark__1}
			\nonumber g'(x)&=\alpha^{\downarrow}(x+L^{\min}) = (h \circ g^{\downarrow})^{\downarrow}(x+L^{\min})\\
			&=\left((h \circ g^{\downarrow})^{\uparrow}\right)^-(x+L^{\min}).
		\end{align}
		By \eqref{eq:comp-up} and since $h^\uparrow(x)= \left[x-L^{\max}\right]^+$, we have:
		\begin{align}
			\nonumber g'(x)&=\left((g^{\downarrow})^{\uparrow} \circ h^{\uparrow}\right)^-(x+L^{\min})\\
			 &= \left((g^{\downarrow})^{\uparrow}\right)^-([x-L^{\max}]^++L^{\min}).
		\end{align}
		If $g$ is left continuous, by Lemma~\ref{lem:upper-lower-conseq}, we have:
		\begin{align}\label{eq:prop2_1_remark__2}
			g'(x)&= g([x-L^{\max}]^++L^{\min}).
		\end{align}
		The above equation shows that applying item 2 and then item 1 does not give the same $g$-regularity as the initial one, i.e., $g\neq g'$, except from the case that $L^{\max} = L^{\min}$ (when all packets have the same length).

\subsection{Proof of Proposition \ref{pro:rate2reg}} \label{proof:rate2reg}
1)	
According to the min-plus representation of packet-level arrival-curve in Eq. \eqref{eq:pcr_minplus}, for any packets $m,n$ with $m \leq n$, we have: 
	\begin{equation}
	n-m+1 \leq \alphapcr{}^+(E_n-E_m).
	\end{equation}
	Now let us multiply both sides of the inequality by $L^{\max}$:
	\begin{equation}\label{eq:prop_3_1}
		(n-m+1)L^{\max} \leq L^{\max} \alphapcr{}^+(E_n-E_m).
	\end{equation}
	For all packet indices $i$, it holds that $l_i \leq L^{\max}$. Thus,
	\begin{equation}
	\sum_{i=m}^{n} l_i \leq L^{\max} \alphapcr{}^+(E_n-E_m).
	\end{equation}
	According to \eqref{eq:arrival_minplus}, the flow conforms to an arrival curve \mbox{$\alpha = L^{\max} \alphapcr{}$}.
	
	To obtain $g$-regularity, we use Eq. \eqref{eq:prop_3_1} as well as the fact that $l_i \leq L^{\max}$ for all packets $i$, and we have:
	\begin{equation}
	\sum_{i=m}^{n-1} l_i + L^{\max} \leq L^{\max} \alphapcr{}^+(E_n-E_m).
	\end{equation}
	Now, we divide the both sides by $L^{\max}$ and set $x:=\sum_{i=m}^{n-1} l_i$. Then,
	\begin{equation}
	\frac{x}{L^{\max}}+ 1 \leq \alphapcr{}^+(E_n-E_m).
	\end{equation}
	Using \eqref{eq:liebeherr_P8} and then \eqref{eq:jylb_theory_1}:
	\begin{equation}
	E_n-E_m \geq \alphapcr{}^{\downarrow}(\frac{x}{L^{\max}}+1) = g(x).
	\end{equation}

2) First we show that we can derive \eqref{eq:acapp} using \eqref{eq:gregapp}.
By item 1, the flow with packet-level arrival-curve $\alphapcr{}$ conforms to $g$-regularity with $g=\alphapcr{}^{\downarrow} \circ f$, where \mbox{$f(x)=\frac{x}{L^{\max}}+1$}. The flow also conforms to a bit-level arrival-curve $\alpha'$ that is derived by applying item 2 of Proposition~\ref{pro:arr-g} to \eqref{eq:gregapp}, i.e.,  $\alpha'(t)=g^{\downarrow}(t) + L^{\text{max}}$; then,
	\begin{align}\label{eq:prop3_1_remark}
	 \nonumber\alpha'(t) &= g^{\downarrow}(t) + L^{\max} = \left(\alphapcr{}^{\downarrow} \circ f \right)^\downarrow (t)+L^{\max} \\
	 &= \left(\left(\alphapcr{}^{\downarrow} \circ f \right)^\uparrow\right)^- (t)+L^{\max} 
	\end{align}
where the last equality is obtained by \eqref{eq:jylb_theory_3}. Now, using \eqref{eq:comp-up} and left continuity of $\alphapcr{}^{\downarrow}$:
\begin{align}
    \alpha'(t)=\left(f^\uparrow \circ \left(\alphapcr{}^{\downarrow}\right)^\uparrow\right)^-(t)+L^{\max},
\end{align}
Since $f^\uparrow(t)=[t L^{\max}-L^{\max}]^+$, we have:
	\begin{align}
		\alpha'(t)&= \left[L^{\max} \left(\left(\alphapcr{}^{\downarrow}\right)^\uparrow\right)^-(t) - L^{\max} \right]^+ + L^{\max}.
	\end{align}
	Since $\alphapcr{}$ is left continuous, by Lemma~\ref{lem:upper-lower-conseq}:
    \begin{align}
        \alpha'(t) = L^{\max} \alphapcr{}(t) - L^{\max} + L^{\max}.
    \end{align}
    Note that \mbox{$\alphapcr{}^+(t)\geq 1$}. Thus $\alpha'(t)=\alpha(t)$ given by \eqref{eq:acapp}.
	
Second we show that \eqref{eq:acapp} does not give \eqref{eq:gregapp}.
By item 1, the flow with packet-level arrival-curve $\alphapcr{}$ conforms to bit-level arrival-curve with $\alpha=f' \circ \alphapcr{}$, where $f'(t)=t L^{\max}$.
In addition, the flow also conforms to a $g'$-regularity constraint that derives by applying item of Proposition~\ref{pro:arr-g} to \eqref{eq:acapp}:
	\begin{align}
	\nonumber g'(x)&= \alpha^{\downarrow}(x+L^{\min}) =\left(f' \circ \alphapcr{}\right)^\downarrow(x+L^{\min}) \\
	&=\left(\alphapcr{}^\downarrow \circ {f'}^\downarrow\right)(x+L^{\min}),
	\end{align}
where the last equality is obtained by using \eqref{eq:comp-down} and continuity of $f'$. Now,  since ${f'}^\downarrow(x)=\frac{x}{L^{\max}}$, we have:
	\begin{align}
	g'(x) &= \alphapcr{}^\downarrow \left(\frac{x+L^{\min}}{L^{\max}}\right).
	\end{align}
	When all packets are of the same size ($L^\max=L^\min$), $g'=g$. When $L^\min<L^\max$, since $\alphapcr{}^{\downarrow}$ is a wide-sense increasing function, we have $g' \leq g$ and $g'\neq g$, i.e. $g'$ is a weaker constraint than $g$ given in \eqref{eq:gregapp}.

\subsection{Proof of Lemma~\ref{lem:queuing-delay}}\label{proof:queuing-delay}
    Let $n$ be the index of the packet of interest with length $l_n$. Using Lemma 1 of \cite{mohammadpour_improved_2019}, there exists an $m\leq n$ such that:
	\begin{align}\label{eq:queuing-1}
	\beta(Q_n-A_m) \leq \sum_{k=m}^{n-1} l_k,
	\end{align}
	where $Q_n$ is the beginning of transmission of packet $n$.
	Using \eqref{eq:liebeherr_P6}, Eq. \eqref{eq:queuing-1} gives:
	\begin{equation}
	Q_n-A_m \leq \beta^{\uparrow}\Big(\sum_{k=m}^{n-1} l_k\Big).
	\end{equation}
	Therefore $Q_n$ satisfies,
	\begin{align}
	Q_n \leq \max_{m\leq n} \Bigg\{A_m+\beta^{\uparrow}\Big(\sum_{k=m}^{n-1} l_k\Big)\Bigg\}.
	\end{align}
	Since $\sum_{k=m}^{n-1} l_k \leq w(A_n-A_m)$, we have:
	\begin{align}\label{eq:scheduler_del_6}
    	Q_n \leq \max_{m\leq n} \Bigg\{A_m+\beta^{\uparrow}\Big(w(A_n-A_m)\Big)\Bigg\}.
	\end{align}
	By defining $t \eqdef A_n-A_m \geq 0$, we further obtain,
	\begin{align}\label{eq:queuing-2}
	Q_n - A_n \leq \sup_{t \geq 0} &\left\{-t+\beta^{\uparrow}\left(w\left(t\right)\right)\right\}.
	\end{align}
	Applying \eqref{eq:jylb_theory_2} to \eqref{eq:queuing-2}:
	\begin{align}\label{eq:queuing-3}
	Q_n - A_n \leq \sup_{t \geq 0} &\left\{(\beta^{\downarrow})^+\left(w\left(t\right)\right)-t\right\}.
	\end{align}
	Next, we use Lemma 2 of \cite{mohammadpour_improved_2019}, which implies that for, $f\in\Fincz$,  $$\sup_{t \geq 0}\left\{f^+(t)-Rt\right\} = \sup_{t \geq 0}\left\{f(t)-Rt\right\}.$$
	We apply this result to \eqref{eq:queuing-3} by setting $\beta^{\downarrow}~=~f$, and therefore:
	\begin{align}\label{eq:scheduler_del_13}
	Q_n - A_n \leq \sup_{t \geq 0} &\left\{\beta^{\downarrow}\left(w\left(t\right)\right)-t\right\} = h(w,\beta),
	\end{align}
	which completes the proof.


\subsection{Proof of Lemma~\ref{lem:delay-alpha-g}}\label{proof:delay-alpha-g}
(i)~Let us remind that $l_{k}$ and $A_k$ are the length and the arrival time of the $k^{th}$ packet  with $k=1,2...$ (Section~\ref{sec:sys}). Let $n$ be the index of the packet of interest belonging to flow $1$ with length $l_n$, $l_n=l$. The sum of all packets can be split in two parts, one with packets belonging to flow $1$ and the one with packets belonging to flow $2$. Let $F(k)$ be the flow id of $k^{th}$ packet, then for any $0\leq m \leq n $, we have:
\begin{align}\label{eq:arrg-delay-0}
    \sum_{k=m}^{n-1} l_k = \sum_{k=m}^{n-1} 1_{\{F(k)= 1\}}l_k + \sum_{k=m}^{n-1} 1_{\{F(k)= 2\}}l_k.
\end{align}
For flow $1$ with g-regularity constraint,
by applying \eqref{eq:liebeherr_P6} to \eqref{eq:g_maxplus}, we have:
\begin{align}\label{eq:g-maxplus-delay-arrg}
    \sum_{k=m}^{n-1} 1_{\{F(k)= 1\}}l_k \leq g^{\uparrow}(A_n - A_m),
\end{align}
For flow $2$ with bit-level arrival-curve, using \eqref{eq:arrival_minplus} we have:
\begin{align}\label{eq:arrg_del_8a}
    \sum_{k=m}^{n} 1_{\{F(k)= 2\}}l_k \leq {\alpha}^+(A_n-A_m).
\end{align}
\noindent Note that since packet $n$ belongs to flow $1$, $\sum_{k=m}^{n} 1_{\{F(k)= 2\}}l_k=\sum_{k=m}^{n-1} 1_{\{F(k)= 2\}}l_k$. Therefore, together with \eqref{eq:g-maxplus-delay-arrg} and \eqref{eq:arrg_del_8a}:
\begin{align}\label{eq:arrg_del_2a}
    \nonumber \sum_{k=m}^{n-1} &1_{\{F(k)=1\}}l_k + \sum_{k=m}^{n-1} 1_{\{F(k)=2\}}l_k \\
    & \leq g^{\uparrow}(A_n - A_m) + {\alpha}^+(A_n-A_m).
\end{align}
Using \eqref{eq:arrg_del_2a} in \eqref{eq:arrg-delay-0}:
	\begin{equation}\label{eq:scheduler_del_4}
	\sum_{k=m}^{n-1} l_k \leq w(A_n-A_m),
	\end{equation}
	where the function $w: \mathbb{R}^+\rightarrow\mathbb{R}^+$ is $w= g^{\uparrow}+ \alpha^+$. Then by Lemma \ref{lem:queuing-delay} for packet $n$, we have:
\begin{align}
    Q_n - A_n \leq h(w,\beta),
\end{align}
where $Q_n$ is start of transmission of packet $n$. Since the transmission time for the packet of interest, $n$, is $D_n - Q_n=\frac{l}{c}$, the delay bound is:
\begin{align}
    D_n - A_n &= D_n - Q_n +Q_n-A_n \leq h(w,\beta) +\frac{l}{c},
\end{align}
which concludes the proof for item 1. Now, since $l\leq L_1^{\max}$, $h(w,\beta) +\frac{L_1^{\max}}{c}$ is a delay bound for flow $1$ which completes the proof.

\subsection{Proof of Theorem \ref{thm:delay_pcr_tight}}\label{proof:delay_pcr_tight}
Let us first define the function $w$ as:
    \begin{align}
      w(t) = \sum_{u=1}^{M} L_u^{\max}\alphapcr{u}(t),
    \end{align}
    The proof is in two steps: first, we construct a simulation trace; second, we verify its properties.
    
    \textbf{Step 1}.
	We start the construction of a simulation trace. 
	
	(a) We determine the smallest time instant $t'$ with the following property:
	\begin{align}\label{eq:tightness_1}
	\nonumber \beta^\downarrow\left(w(t')-L_1^{\max}\right)-t' &= \sup_{t \geq 0}\left\{\beta^\downarrow\left(w(t)-L_1^{\max}\right)-t\right\}\\
	&=h(w-L_1^{\max},\beta).
	\end{align}
	In fact, the time $t'$ will be the time at which the packet of interest arrives at the system and experiences the worst-case delay.
	
	(b) Now, we generate the packet sequence for the $M$ flows.	Flow $1$ has $n_1={\alphapcr{1}}^+(t')$ packets presented as a pair $(A^1,\mathcal{L}^1)$ where $A^1=(A^1_1,A^1_2,\dots,A^1_{n_1})$ is the sequence of packet arrival times and $\mathcal{L}^1 = (l^1_1,l^1_2,\dots,l^1_{n_1})$ is the packet length sequence,  defined for
	 all $i~\in~\{1,\dots,n_1\}$ by:
	\begin{align}\label{eq:tightness_flow1_config}
	& l^1_i = L_1^{\max},\\
	\nonumber& A^1_i = \gamma+\inf\{s\geq 0~|~ \alphapcr{1}(s) \geq i\}=\gamma+\alphapcr{1}^{\downarrow}(i),
	\end{align}
	where $\gamma=t'-\alphapcr{1}^{\downarrow}(n_1)$. Lemma~\ref{lem:tightness2} shows that \mbox{$\alphapcr{1}^{\downarrow}(n_1)\leq t'$} and hence $\gamma\geq 0$.
	The aforementioned packet sequence indicates that the packets have maximum length and the packet arrival is greedy, starting at time $\gamma$, and the last packet (packet of interest) arrives at $A^1_{n_1}= t'$.
	
	Any other flow $f$, $f\neq1$, has $n_f={\alphapcr{f}}^+(t')$ packets presented as a pair $(A^f,\mathcal{L}^f)$ where $A^f=(A^f_1,A^f_2,\dots,A^f_{n_f})$ is the packet arrival sequence and $\mathcal{L}^f = (l^f_1,l^f_2,\dots,l^f_{n_f})$ is the packet length sequence, that are defined for all $j~\in~\{1,\dots,n_f\}$ as:
	\begin{align}\label{eq:tightness_flow2_config}
	\nonumber& l^f_j = L_f^{\max},\\
	&A^f_j = \inf\{s\geq 0~|~ {\alphapcr{f}}(s) \geq j\}= {\alphapcr{f}}^{\downarrow}(j).
	\end{align}
	The aforementioned packet sequence indicate that the packets have maximum length and their arrival is greedy starting at time \mbox{$A^f_1={\alphapcr{f}}^{\downarrow}(1)=0$}. With the above arrival construction, we have the cumulative input packet-count function as:
	\begin{align}
	    N_u(t) = 
	    \begin{cases}
          \sum_{\substack{i=1}}^{n_u} 1_{A^u_i< t} & \text{$t\leq t'$}\\
          n_u & \text{$t> t'$}
        \end{cases}  
        ~,~\forall u=1,2,\dots,M.
	\end{align}
	
	Now let us merge the two packet sequences to express the total traffic. We define the pair $(A,\mathcal{L})$ with arrival sequence $A=(A_1,A_2,\dots,A_n)$ and lengths sequence $\mathcal{L}=(l_1,l_2,\dots,l_n)$ of total packets and $n=\sum_{k=1}^{M}n_k$, i.e., \mbox{$(A,\mathcal{L}) = \bigcup_{k=1}^{M}(A^k,\mathcal{L}^k)$}. Lemma~\ref{lem:tightness2} indicates that for any $f\neq 1$:
	$$
	A^f_{n_f} \leq t'=A^1_{n_1}.
	$$
	In the case $A^f_{n_f}=A^1_{n_1}=t'$, assume that the last packet of flow $1$ is enqueued after the last packet of other flows; hence, $A_n=A^1_{n_1}=t'$ and $l_n=L^{\max}_1$.
	
	The cumulative input function, $I(t)$ is shown as the green line in Fig. \ref{fig:delay_pcr_tight} that is obtained as:
	\begin{align}
	\nonumber I(t) &= \sum_{u=1}^{M}L_k^{\max} N_u(t)=\begin{cases}
          \sum_{\substack{k=1}}^{\infty} L_k^{\max}1_{\{A_k< t\}} & t\leq t'\\
          \sum_{u=1}^{M}n_u L_u^{\max}=w^+(t) & t> t'
        \end{cases} 
	\end{align}

	\begin{figure}[t]
		\centering
		\includegraphics[width=0.9 \linewidth]{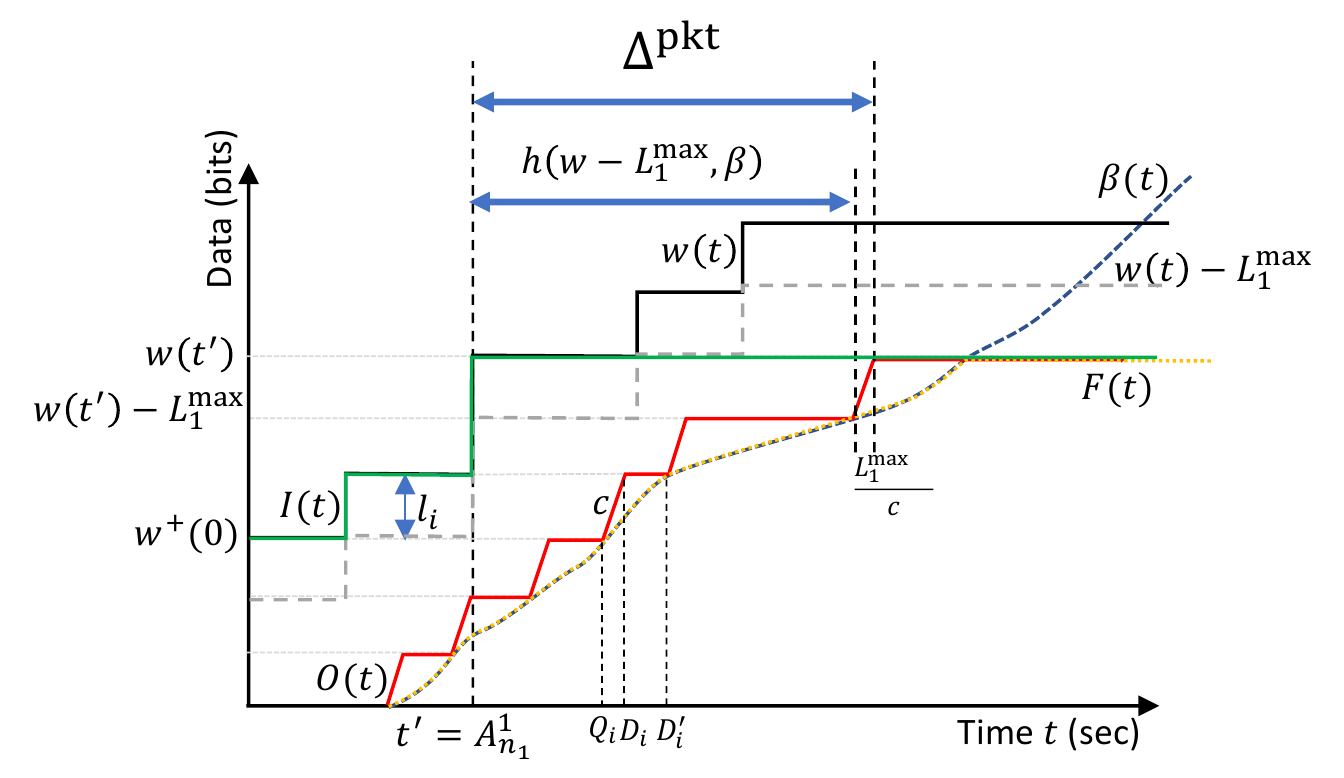}
		\caption{The execution trace used in the proof of Theorem~\ref{thm:delay_pcr_tight}. The delay of the packet with length $L^{\max}_1$ that arrives at time $t'$ is $\Delta^{\mathrm{pkt}}$.}
		\label{fig:delay_pcr_tight}
	\end{figure}
	
	(c) For the output, we first construct the fluid output curve $F(t)$ (orange dotted-line in Fig. \ref{fig:delay_pcr_tight}) given by
    \begin{align} \label{eq:output_tight}
    	F(t) = \inf_{0\leq s \leq t}\left\{I(s)+ \beta(t-s)\right\},
	\end{align}
    so that the service curve property would be automatically satisfied if we would let the output cumulative function be $O(t)=F(t)$. However, we cannot take $O(t)=F(t)$ because $F(t)$ does not satisfy the condition that packet transmission is at rate $c$. In order to obtain the output function $O(t)$, we first observe the start and end of transmission time of a packet $l_i$, i.e.,  $Q_i$ and $D'_i$ as:
    \begin{align}
        \nonumber D'_i &= \inf\left\{s\geq0 : F(s) \geq I^+(A_i)\right\} = F^{\downarrow}\left(I^+(A_i)\right),\\
        \nonumber Q_i &= \inf\left\{s\geq0 : F(s) \geq I^+(A_i)-l_i\right\} = F^{\downarrow}\left(I^+(A_i)-l_i\right).
    \end{align}
    In the above definitions, we have $D'_i = Q_{i+1}$ as \mbox{$I^+(A_{i+1})=I^+(A_i)+l_i$} by definition of $I$.
    For the output function $O(t)$, we keep the same time $Q_i$ for the start of transmission of packet $i$, but, we let the transmission finish at time $D_i = Q_i + \frac{l_i}{c}$. Observe that:
    \begin{align}
        D'_i - D_i = F^\downarrow\left(I^+(A_i)\right) - \left(F^\downarrow\left(I^+(A_i)-l_i\right) + \frac{l_i}{c}\right).
    \end{align}
    By Lemma \ref{lem:lipconv}, $F(t)$ is $c$-Lipschitz. Then using Lemma \ref{lem:pseudo_inverse_Lipschitz},
    \begin{align}
        D'_i - D_i \geq \frac{l_i}{c} - \frac{l_i}{c} = 0.
    \end{align}
    Then, more precisely $\forall i=1,\dots,n_1+n_2$ and \mbox{$\forall t\in[Q_i,D'_{i}]$}, $O(t)$ is:
    \begin{align}
        O(t)=
        \begin{cases}
                  c(t-Q_i)+F(Q_i) & Q_i\leq t< D_i,\\
                  F(Q_i) + l_i & D_i\leq t\leq D'_{i}.
        \end{cases}
    \end{align}
    $O(t)$ is shown with red line in Fig. \ref{fig:delay_pcr_tight}.
    
    \textbf{Step 2.} 
	We verify that all requirements in the theorem are satisfied. 
	First we show that the service curve property holds; to do this, since $F$ satisfies the service curve property, it is sufficient to show that $O(t)\geq F(t): \forall t\in[Q_i,D'_{i}]$ for any \mbox{$i=1,\dots,n$}. Using Lipschitz continuity of $F$ for \mbox{$t\in[Q_i,D_i)$},
	\begin{align}
	    F(t) \leq F(Q_i) + c(t-Q_i) = O(t).
	\end{align}
	For $t\in[D_i,D'_{i}]$, by construction:
	\begin{align}
	    F(t) \leq F(Q_i) + l_i = O(t).
	\end{align}
	The above equations imply $O(t) \geq F(t), \forall t\geq 0$; therefore, the service curve property is satisfied when the output is $O(t)$. 
	
	Moreover, by construction the system is FIFO, the input is packetized and packet transmission occurs at rate $c$. We need to prove that the input conforms to the packet-level arrival-curve. For any flow $f$, $f\neq 1$, consider two time instants $s,t\geq 0$, $s\leq t$. 
	
	\noindent If $t=s=0$, then $N_f(t)-N_f(s)=\alphapcr{f}(t-s)=0$. 
	
	\noindent If $0<t\leq t'$, there exist a packet index $m'$ where \mbox{$t\in(A^f_{m'},A^f_{m'+1}]$}. Then by Lemma \ref{lem:tightness-arrival} with $k=0$, we have \mbox{$m'=\alphapcr{f}(t)$}. Now if $s=0$, we have 
	$$N_f(t)-N_f(s) = N_f(A^f_{m'+1})-N_f(0) = m'=\alphapcr{f}(t);$$
	otherwise, there exists a packet index $m$ ($m\leq m'$), where $s\in(A^f_m,A^f_{m+1}]$. Then:
	\begin{align}
	    \nonumber N_f(t)-N_f(s) = N_f(A^f_{m'+1})-N_f(A^f_{m+1}) = m' - m.
	\end{align}
	By Lemma \ref{lem:tightness-arrival} with $k=0$, $m'=\alphapcr{f}(t)$ and $m=\alphapcr{f}(s)$. Therefore:
	\begin{align}
	    \nonumber N_f(t)-N_f(s) = \alphapcr{f}(t)-\alphapcr{f}(s) \leq \alphapcr{f}(t-s),
	\end{align}
	where the last inequality is due to sub-additivity of $\alphapcr{f}$.
	
	\noindent If $t> t'$, by construction, $N_f(t)=n_f=\alphapcr{f}^+(t')$. Now, for $s=0$, we have 
	$$N_f(t)-N_f(s) = n_f={\alphapcr{f}}^+(t') \leq \alphapcr{f}(t); \forall t>t'.$$
	For $0<s\leq t'$, there exists a packet index $m$ ($m\leq m'$), where $s\in(A^f_m,A^f_{m+1}]$. Then by Lemma \ref{lem:tightness-arrival} with $k=0$, we have \mbox{$m=\alphapcr{}(s)$}. Therefore, 
	\begin{align}
	    \nonumber N_f(t)-N_f(s) &= n_f-m ={\alphapcr{f}}^+(t') - \alphapcr{f}(s) \\
	    \nonumber &\leq {\alphapcr{f}}(t-s); \forall t>t',
	\end{align}
	where the last inequality is due to sub-additivity of ${\alphapcr{f}}$.
	
	\noindent For $t'<s \leq t$, $N_f(s)=n_f$. Then,
	$$N_f(t)-N_f(s) = n_f-n_f=0\leq \alphapcr{f}(t-s), $$
	which shows flow $f$ conforms to the packet-level arrival-curve. 
	
	For flow $1$, similarly to the above computation and using Lemma~\ref{lem:tightness-arrival} with $k=\gamma$, we obtain
	$$N_1(t)-N_1(s) \leq \alphapcr{}(t-s); \forall s\leq t, \forall t\geq 0,$$ 
	which shows flow $1$ conforms to packet-level arrival-curve.
	
	Last, we show that packet $n$ achieves the delay bound. We have for packet $n$, $A_n = t'$; then:
	\begin{align}
	Q_n &=  F^\downarrow(I^+(t')-l_n) = F^\downarrow(w^+(t')-L^{\max}_1).
	\end{align}
	 Furthermore, $D_n=Q_n+\frac{l_n}{c} = Q_n+\frac{L_1^{\max}}{c}$, therefore,
	\begin{align}
	\nonumber D_n -A_n&= F^\downarrow(w^+(t')-L^{\max}_1) -t'+\frac{L_1^{\max}}{c}.
	\end{align}
	By definition of $F$ in \eqref{eq:output_tight}, we have $F\leq \beta$. Then by \cite[Lemma 10.1]{liebeherr_duality_2017}, $F^\downarrow \geq \beta^{\downarrow}$. Hence:
	\begin{align}\label{eq:delay-tight-final}
	\nonumber D_n -A_n &\geq \beta^\downarrow(w^+(t')-L^{\max}_1) -t'+\frac{L_1^{\max}}{c}\\
	&= h(w^+ -L^{\max}_1,\beta) +\frac{L_1^{\max}}{c}=\Delta^{\pkt}.
	\end{align}
	By Theorem \ref{thm:delay_pcr}, we have $D_n -A_n \leq \Delta^{\pkt}$; then together with the above equation, we have $D_n -A_n = \Delta^{\pkt}$.

\begin{lemma}\label{lem:tightness2}
	Consider a wide-sense increasing function $f: \mathbb{R^+}\rightarrow\mathbb{Z^+}$, a positive real value $t_0$, a positive integer $n=f^+(t_0)$, and a value $x=f^{\downarrow}(n)$; then we have $x \leq t_0$.
\end{lemma}
\begin{proof}
	We have:
	\begin{equation}
	x = f^{\downarrow}(n) = f^{\downarrow}(f^+(t_0)).
	\end{equation}
	By \eqref{eq:jylb_theory_1}, $f^{\downarrow}=(f^+)^{\downarrow}$. Therefore, $x=(f^+)^{\downarrow}(f^+(t_0))$;
	then using Property (P1) of \cite[Chapter 10.1]{liebeherr_duality_2017} with $F=f^+$, we have $(f^+)^{\downarrow}(f^+(t_0))~\leq~t_0$ which concludes the proof.
\end{proof}

\begin{lemma}\label{lem:lipconv}
    Consider two functions $f,g$ where $f$ is $L$-Lipschitz and $g\geq0$. Then $z(t) = \inf_{0\leq s\leq t}\{g(s)+f(t-s)\}$ is $L$-Lipschitz.
\end{lemma}
\begin{proof}
    Let us define the set of functions $w_s(t)$ with constant $s\in\R^+$, as:
    \begin{align}
        w_s(t) = g(s) + f(t-s).
    \end{align}
    Then, $z(t)=\inf_{0\leq s\leq t}\{w_s(t)\}$. First we prove that $w_s$ is \mbox{$L$-Lipschitz} for any $s \in \R^+$. For any $t_1,t_2\in \R^+$ and $t_2\geq t_1$:
    \begin{align}\label{eq:lem-lip-inf}
        \nonumber |w_s(t_2)-&w_s(t_1)| = |g(s) + f(t_2-s)- g(s) - f(t_1-s)| \\
        &= |f(t_2-s)-f(t_1-s)| \leq L |t_2-t_1|.
    \end{align}
    The last inequality is obtained as $f$ is $L$-Lipschitz. Now we prove the lemma. By \eqref{eq:lem-lip-inf}, we have for any $s\in \R^+$:
    \begin{align}\label{eq:lem-lip-inf-2}
        w_s(t_1) - L |t_2-t_1| \leq w_s(t_2) \leq w_s(t_1) + L |t_2-t_1|.
    \end{align}
    Using the left inequality, we have:
    \begin{align}
        \forall s\in \R^+~:~\inf_{0\leq u\leq t_1}\left\{w_u(t_1) - L |t_2-t_1|\right\} \leq w_s(t_2).
    \end{align}
    that gives $z(t_1)- L |t_2-t_1|\leq w_s(t_2)$; therefore,
    \begin{align}\label{eq:lem-lip-inf-3}
        z(t_1)- L |t_2-t_1| \leq \inf_{0\leq s\leq t_2}\left\{w_s(t_2)\right\}=z(t_2)
    \end{align}
    Using the right inequality in \eqref{eq:lem-lip-inf-2}, we have:
    \begin{align}
        \nonumber\inf_{0\leq s\leq t_2} \{w_s(t_2)\} &\leq \inf_{0\leq s\leq t_2}\left\{w_s(t_2) + L |t_2-t_1|\right\}\\
        &\leq \inf_{0\leq s\leq t_1}\left\{w_s(t_2)\right\}+ L |t_2-t_1|,
    \end{align}
    that gives $z(t_2)\leq z(t_1)+ L |t_2-t_1|$. Then, together with \eqref{eq:lem-lip-inf-3}:
    \begin{align}
        z(t_1)- L |t_2-t_1| \leq z(t_2)\leq z(t_1)+ L |t_2-t_1|.
    \end{align}
    Hence, $|z(t_2) - z(t_1)| \leq L |t_2-t_1|$, which concludes the proof.
\end{proof}

\begin{lemma}\label{lem:tightness-arrival}
    Consider a left-continuous function $\alpha:\R^+\rightarrow \N$ and a constant $k\geq 0$. Let $A=(A_1,A_2,\dots)$ be a sequence where $A_i=k+\alpha^{\downarrow}(i)$. Assume a positive integer $m$ and a time instant $t$, such that $t\in(A_m,A_{m+1}]$; then we have \mbox{$m=\alpha(t-k)$}.
\end{lemma}
\begin{proof}
    We prove $\alpha(t-k)\geq m$ and $\alpha(t-k)\leq m$. We have \mbox{$t > A_m =k+\alpha^{\downarrow}(m)$}; therefore, \mbox{$t-k > \alpha^{\downarrow}(m)$}. Then, by \cite[Proposition 6]{boyer_common}, we have $\alpha(t-k) \geq m$.
    
    \noindent Also, since $t\leq A_{m+1}$, for any $\varepsilon>0$ we have:
    \begin{align}
        t - \varepsilon < A_{m+1} = k+\alpha^{\downarrow}(m+1).
    \end{align}
    Therefore,
    \begin{align}
        t -k- \varepsilon < \alpha^{\downarrow}(m+1).
    \end{align}
    Then, by \cite[Proposition 6]{boyer_common}, $\alpha(t-k-\varepsilon) < m+1$. Hence:
    \begin{align}
        \lim_{\varepsilon\rightarrow0}\alpha(t-k-\varepsilon) < m+1.
    \end{align}
    As $\alpha$ is left-continuous, we have $\alpha(t-k-\varepsilon)=\alpha(t-k)<m+1$. Since $\alpha(t-k)\in \N$, $\alpha(t-k)\leq m$.
\end{proof}

\subsection{Proof of Theorem \ref{thm:delay_tight_fixed_interval}}\label{proof:delay_tight_fixed_interval}
As discussed in Section~\ref{sec:pktarr}, the flows can have sliding-interval and fixed-interval regulation constraints; let $F$ and $S$ be respectively the sets of flows with fixed interval and sliding interval. 
By \eqref{eq:cpcr}, if a flow $i$ has the sliding interval ($\tau_i,K_i$) constraint, this is equivalent to have the following packet-level arrival-curve:
\begin{align}\label{eq:arr-slid-tight}
\alphapcr{i}(0)=0,~\alphapcr{i}(t) = K_i\ceil{\frac{t}{\tau_i}},~t>0.
\end{align}
By \eqref{eq:fpcr}, if a flow $i$ has the fixed interval ($\tau_i,K_i$), the constraint implies the following packet-level arrival-curve:
\begin{align}
    \alphapcr{j}(0)=0,~\alphapcr{j}(t) = K_j\ceil{\frac{t}{\tau_j}}+K_j,~t>0.
\end{align}
Then the delay bound obtain by Theorem~\ref{thm:delay_pcr} for flow $1$ is:
\begin{align}
    \nonumber \Delta^\pkt=&h(\sum_{u=1}^{N}L_u^{\max}K_u\ceil{\frac{t}{\tau_u}}+\sum_{u=1}^{N}L_u^{\max}K_u1_{\{u\in F\}}-L_1^{\max},\beta)\\
    &+\frac{L_1^{\max}}{c}.
\end{align}
We want to construct a simulation trace where a packet of flow $1$ experiences a delay arbitrarily close to $D$.
To this end, for flows in $S$, we use \eqref{eq:arr-slid-tight} to construct the input packet sequence as it is equal to the sliding interval interpretation. 
For any other flow $j$ with fixed interval ($\tau_j,K_j$), we use the following packet-level arrival-curve:
\begin{align}\label{eq:arr-fix-tight-ep}
    \nonumber \alphapcr{j}^\varepsilon(0)&=0,\\
    \alphapcr{j}^\varepsilon(t)&=K_j\ceil{\frac{[t-\varepsilon]^+}{\tau_j}}+K_j : t>0, \forall\varepsilon\in(0,\min_{u}\{\tau_u\}).
\end{align}
Lemma~\ref{lem:arr-fixed-interval-greedy} shows that a greedy packet-sequence of $\alpha_\varepsilon^j$, starting at  \mbox{$t_0=\max_u\{\tau_u\}-\varepsilon$}, conforms to  Fixed interval ($\tau_j,K_j$).

\noindent Now, let:
\begin{align}\label{eq:def-w-tsn}
      \nonumber w_\varepsilon(t) &= \sum_{u=1}^{N}L_u^{\max}\alphapcr{}^u(t)1_{\{u\in S\}}+\sum_{u=1}^{N}L_u^{\max}\alphapcr{u}^\varepsilon(t)1_{\{u\in F\}}\\
      \nonumber &=\sum_{u=1}^{N}L_u^{\max}K_u\ceil{\frac{t}{\tau_u}}1_{\{u\in S\}}\\
      &+\sum_{u=1}^{N}L_u^\max K_u\ceil{\frac{[t-\varepsilon]^+}{\tau_u}}1_{\{u\in F\}}+\sum_{u=1}^{N}K_u1_{\{u\in F\}}.
\end{align}
The tightness scenario follows the same as Theorem~\ref{thm:delay_pcr_tight} by generating a greedy packet-sequence for every flow using the packet-level arrival-curves in \eqref{eq:arr-fix-tight-ep} for fixed-interval flows and \eqref{eq:arr-slid-tight} for sliding-interval flows. To create a feasible greedy packet-sequence, we shift the start of the simulation by $t_0$ (as mentioned earlier, this guarantees the existence of greedy packet-sequence for the fixed-interval flows).
Therefore, a packet of a flow of interest, i.e., flow $1$, experiences a delay of $d_\varepsilon=h(w_\varepsilon -L^{\max}_1,\beta) +\frac{L_1^{\max}}{c}$ as shown in \eqref{eq:delay-tight-final}. By definition of $w_\varepsilon$ in \eqref{eq:def-w-tsn}, $d_\varepsilon$ is dependent on $\varepsilon$. We show next that $d\geq \Delta^\pkt-\varepsilon$.

To this end, let us define the following auxiliary functions:
\begin{align}
    \nonumber f(t) &=\sum_{u=1}^{N}L_u^\max K_u\ceil{\frac{t}{\tau_u}}1_{\{u\in F\}},~t>0;~f(t)=0, t\leq 0,\\
     g(t) &= \sum_{u=1}^{N}L_u^{\max}K_u\ceil{\frac{t}{\tau_u}}1_{\{u\in S\}}+\sum_{u=1}^{N}K_u1_{\{u\in F\}}-L^{\max}_1.
\end{align}
Then we have 
\begin{align}
    \nonumber w_\varepsilon(t)&=f(t-\varepsilon)+g(t)+L^{\max}_1\\
    \Delta^\pkt&=h(f+g,\beta)+\frac{L_1^{\max}}{c}.
\end{align}
By Lemma~\ref{lem:h-shift}, $h(w_\varepsilon -L^{\max}_1,\beta)\geq h(f+g,\beta)-\varepsilon$; therefore,
\begin{align}
    d_\varepsilon \geq h(f+g,\beta) +\frac{L_1^{\max}}{c}-\varepsilon = \Delta^\pkt-\varepsilon.
\end{align}
Moreover, by Theorem~\ref{thm:delay_pcr}, $d_\varepsilon\leq \Delta^\pkt$; hence,
$$\Delta^\pkt-\varepsilon\leq d_\varepsilon\leq \Delta^\pkt.$$
Finally, when \mbox{$\varepsilon\rightarrow 0$}, $d_\varepsilon$ gets arbitrary close to $\Delta^\pkt$.

\begin{lemma}\label{lem:arr-fixed-interval-greedy}
    Consider the following packet-level arrival-curve:
    \begin{align}
        \nonumber \alphapcr{}(0)&=0,\\
        \alphapcr{}(t)&=K\ceil{\frac{[t-\epsilon]^+}{\tau}}+K : t>0, \forall\epsilon\in(0,\tau).
    \end{align}
    Then, a greedy packet-sequence of $\alphapcr{}$ at any time \mbox{$t_0\geq \tau-\epsilon$}, conforms to the fixed-interval ($\tau,K$) in \eqref{eq:fixed-interval}.
\end{lemma}
\begin{proof}
    Consider the cumulative packet function $N$ defined as:
    \begin{align}
        N(0)=N(t_0)=0,~~N(t)=\alphapcr{}(t-t_0),~t>t_0.
    \end{align}
     By definition, the above function is greedy at time $t_0$. 
     Now, in \eqref{eq:fixed-interval}, let $\theta=t_0-\tau+\epsilon$. Since \mbox{$t_0\geq \tau-\epsilon$}, we have $0\leq \theta < t_0$; therefore, $N(\theta)=0$. Next, we show that the greedy packet-sequence conforms to the fixed-interval constraint for any time $\geq \theta$. For $i=0$ in \eqref{eq:fixed-interval}:
     \begin{align}
         N(\theta+\tau)-N(\theta)=N(t_0+\epsilon)-0 = \alphapcr{}(\epsilon)=K,
     \end{align}
     and for $i\geq 1$:
     \begin{align}
         \nonumber N(\theta+(i+1)\tau)&-N(\theta+i\tau)=N(t_0+\epsilon+i\tau)\\
         \nonumber&-N(t_0+\epsilon+(i-1)\tau)= \nonumber \alphapcr{}(\epsilon+i\tau)\\
         &-\alphapcr{}(\epsilon+(i-1)\tau)\\
         \nonumber&=(iK+K)-((i-1)K+K)=K.
     \end{align}
     Therefore, $\forall i\in\N$:
     \begin{align}
         N(\theta)=0,~N(\theta+(i+1)\tau)&-N(\theta+i\tau) \leq K,
     \end{align}
     which shows there exists $\theta$ ($=t_0-\tau+\epsilon$) where the cumulative function $N$, as a greedy packet-sequence of $\alphapcr{}$ at $t_0$, conforms to the fixed interval ($\tau,K$) constraint.
\end{proof}

\begin{lemma}\label{lem:h-shift}
    Consider the functions $f,g,\beta\in\Finc$. Let 
    \begin{align}
    f_\epsilon(t)=
    \begin{cases}
              f(t-\epsilon) ~~&t\geq\epsilon,\\
              0 &t< \epsilon.
    \end{cases}    
    \end{align}
    Then \mbox{$h(f_\epsilon+g,\beta)\geq h(f+g,\beta)-\epsilon$}.
\end{lemma}
\begin{proof}
    Therefore,
\begin{align}
    \nonumber h(f_\epsilon+g,\beta) &=\sup_{t \geq 0} \{\beta^{\downarrow}\left(f_\epsilon(t)+g(t)\right)-t\}\\
    \nonumber &\geq \sup_{t \geq \epsilon} \{\beta^{\downarrow}\left(f_\epsilon(t)+g(t)\right)-t\}\\
    &= \sup_{t \geq \epsilon} \{\beta^{\downarrow}\left(f(t-\epsilon)+g(t)\right)-t\}.
\end{align}
Since $g(.)$ is wide-sense increasing, $g(t)\geq g(t-\epsilon),~\forall t\geq \epsilon$. Therefore:
\begin{align}
    \nonumber h(f_\epsilon+g,\beta) &\geq \sup_{t \geq \epsilon} \{\beta^{\downarrow}\left(f(t-\epsilon)+g(t-\epsilon)\right)-t\}\\
    \nonumber&=\sup_{s \geq 0} \{\beta^{\downarrow}\left(f(s)+g(s)\right)-s-\epsilon\}\\
    \nonumber&=\sup_{s \geq 0} \{\beta^{\downarrow}\left(f(s)+g(s)\right)-s\}-\epsilon\\
    &=h(f+g,\beta)-\epsilon,
\end{align}
which concludes the proof.
\end{proof}

\subsection{Proof of Theorem \ref{thm:response_fifo_arrival}} \label{proof:response_fifo_arrival}

(i)~Let us remind that $l_{k}$ and $A_k$ are the length and the arrival time of the $k^{th}$ packet  with $k=1,2...$ (Section~\ref{sec:sys}). Let $n$ be the index of the packet of interest belonging to flow $1$ with length $l_n$, $l_n=l$. The sum of all packets can be split in two parts, one with packets belonging to flow $1$ and the one with packets belonging flow $2$. Let $F(k)$ be the flow id of $k^{th}$ packet, then:
    \begin{align}\label{eq:arr-delay-0}
        \sum_{k=m}^{n-1} l_k = \sum_{k=m}^{n-1} 1_{\{F(k)= 1\}}l_k + \sum_{k=m}^{n-1} 1_{\{F(k)= 2\}}l_k.
    \end{align}
The flow of interest $1$ has a bit-level arrival-curve; using its min-plus representation \eqref{eq:arrival_minplus}, for any $0\leq m \leq n $, we can write:
\begin{align}\label{eq:arcurve1}
    \sum_{k=m}^{n} 1_{\{F(k)=1\}}l_k \leq \alpha^+(A_n - A_m)  .
\end{align}
By excluding the last packet from the left hand-side of \eqref{eq:arcurve1} (note that $l_n = l$), we obtain:
\begin{equation}\label{eq:scheduler_del_1a}
    \sum_{k=m}^{n-1} 1_{\{F(k)=1\}}l_k \leq \alpha^+(A_n - A_m)-l.
\end{equation}
Similarly, flow $2$ has bit-level arrival-curve; thus, using \eqref{eq:arrival_minplus} for any $0\leq m \leq n $, we can write:
\begin{align}\label{eq:scheduler_del_8a}
    \sum_{k=m}^{n-1} 1_{\{F(k)= 2\}}l_k \leq {\alpha'}^+(A_n-A_m).
\end{align}
\noindent Note that since packet $n$ belongs to flow $1$, the above equation conforms the min-plus representation of bit-level arrival-curve in \eqref{eq:arrival_minplus} for flow $2$.	Now, we sum up \eqref{eq:scheduler_del_1a} and \eqref{eq:scheduler_del_8a}:
\begin{align}\label{eq:scheduler_del_2a}
    \nonumber \sum_{k=m}^{n-1} &1_{\{F(k)=1\}}l_k + \sum_{k=m}^{n-1} 1_{\{F(k)=2\}}l_k \\
    & \leq \alpha^+(A_n - A_m) + {\alpha'}^+(A_n-A_m)-l.
\end{align}
Using \eqref{eq:scheduler_del_2a} in \eqref{eq:arr-delay-0}:
	\begin{equation}\label{eq:scheduler_del_4}
	\sum_{k=m}^{n-1} l_k \leq w(A_n-A_m),
	\end{equation}
	where the function $w: \mathbb{R}^+\rightarrow\mathbb{R}^+$ is $w= \alpha^++ \alpha_2^+-l$. Then by Lemma \ref{lem:queuing-delay} for packet $n$, we have:
\begin{align}
    Q_n - A_n \leq h(w,\beta),
\end{align}
where $Q_n$ is start of transmission of packet $n$. Since the transmission time for the packet of interest, $n$, is $D_n - Q_n=\frac{l}{c}$, the delay bound is:
\begin{align}
    D_n - A_n &= D_n - Q_n +Q_n-A_n \leq h(w,\beta) +\frac{l}{c},
\end{align}
which completes the proof.

(ii) We want to compute 
 \begin{align}
\nonumber\sup_{l \in [L_1^{\min}, L_1^{\max}]}  &\left\{\Delta^A(l)\right\}\\
\nonumber &= \sup_l \left[\sup_t \left(\beta^{\downarrow}(\alpha^+ + {\alpha'}^+-l)-t\right)+\frac{l}{c}\right]\\
&= \sup_t \left[\sup_l \left(\beta^{\downarrow}(\alpha^+ + {\alpha'}^+-l)+\frac{l}{c}\right)-t\right].\label{eq:delaybound}
 \end{align}
 Since $\beta$ is $c$-Lipschitz, by Lemma \ref{lem:pseudo_inverse_Lipschitz}, it satisfies:
 \begin{align}
 \beta^{\downarrow}(\alpha^+ + {\alpha'}^+-L^{\min}_1) - &\beta^{\downarrow}(\alpha^+ + {\alpha'}^+-l) \nonumber 
 \\ \geq \frac{1}{c}(l-L^{\min}_1).
 \end{align}
 This gives,
 \begin{align}
  \nonumber \beta^{\downarrow}(\alpha^+ + {\alpha'}^+-l) \leq & \beta^{\downarrow}(\alpha^+ + {\alpha'}^+-L^{\min}_1)\\
  &-\frac{1}{c}(l-L^{\min}_1).
 \end{align}
 By using the last relation in \eqref{eq:delaybound}, we obtain,
 \begin{align}\label{eq:arr-delay-1}
 \nonumber\sup_{l \in [L_1^{\min}, L_1^{\max}]} &\left\{\Delta^A(l)\right\}\\
 \nonumber&= \sup_t \left[\beta^{\downarrow}(\alpha^+ + {\alpha'}^+-L^{\min}_1)-t+\frac{L^{\min}_1}{c}\right]\\
 &=h\left(\alpha^+ + {\alpha'}^+-L^{\min}_1,\beta\right)+\frac{L^{\min}_1}{c}.
 \end{align}
 Since $\beta$ is $c$-Lipschitz continuous and $\alpha,\alpha'$ are left continuous, by \cite[Theorem~5.6]{bouillard_deterministic_2018}, we have:
 \begin{align}
     h\left(\alpha^+ + {\alpha'}^+-L^{\min}_1,\beta\right) = h\left(\alpha + {\alpha'}-L^{\min}_1,\beta\right),
 \end{align}
  which together with \eqref{eq:arr-delay-1} completes the proof.

\subsection{Proof of Proposition \ref{pro:comp-pck}} \label{proof:comp-pck}


From Proposition \ref{pro:rate2reg}, any flow $u$ conform to a bit-level arrival-curve $\alpha_u(t)=\alphapcr{u}(t)L_u^{\max}$. Then by Theorem~\ref{thm:response_fifo_arrival}, the delay bound for a packet with size $l$ of flow $1$ is:
\begin{align}
&\Delta^A(l)=h\left(\sum_{u=1}^{U}\alphapcr{u}L_u^{\max} - l,\beta\right) + \frac{l}{c}.
\end{align}
By \cite[Proposition~5.12]{bouillard_deterministic_2018}, $h$, is monotonically increasing with respect to its first argument; therefore, $\Delta^{pkt}(l)\leq \Delta^A(l)$.
As a result, $\sup_l \Delta^{pkt}(l)\leq \sup_l \left\{\Delta^A(l)\right\}$, i.e., $\Delta^{pkt} \leq \Delta^A$. Note that, if  $L_1^{\max}=L_1^{\min}$ (all packets have the same length), then $\Delta^{pkt}= \Delta^A$. For the general statement, we show a case that when $L_1^{\min}<L_1^{\max}$, the per-flow bound in Theorem~\ref{thm:delay_pcr} strictly improves $\Delta^A$. 

First, for the ease of presentation, let us define \mbox{$w(t)=\sum_{u=1}^{U}\alphapcr{u}L_u^{\max}$}. Next, assume a rate-latency service-curve $\beta(t)=R[t-T]^+$,~$R<c$.
Then, as \mbox{$\beta^{\downarrow}(x)=T+\frac{x}{R}$}, Theorem~\ref{thm:delay_pcr} gives:
\begin{align}\label{eq:pkt-bit-1}
    \nonumber\Delta^{pkt} &= \sup_{t \geq 0} \left\{\beta^{\downarrow}\left(w(t)-L_1^{\max}\right)-t\right\}+\frac{L_1^{\max}}{c}\\
    \nonumber&=\sup_{t \geq 0} \left\{T+\frac{w(t)-L_1^{\max}}{R}-t\right\}+\frac{L_1^{\max}}{c}\\
    &=T-L_1^{\max}\left(\frac{1}{R}-\frac{1}{c}\right)+\sup_{t \geq 0} \left\{\frac{w(t)}{R}-t\right\}.
\end{align}
Using the derived bit-level arrival-curves of flows $1$ and $2$, Theorem~\ref{thm:response_fifo_arrival} gives:
\begin{align}\label{eq:pkt-bit-2}
    \nonumber\Delta^{\mathrm{A}} &= \sup_{t \geq 0} \left\{\beta^{\downarrow}\left(w(t)-L_1^{\min}\right)-t\right\}+\frac{L_1^{\min}}{c}\\
    \nonumber&=\sup_{t \geq 0} \left\{T+\frac{w(t)-L_1^{\min}}{R}-t\right\}+\frac{L_1^{\min}}{c}\\
    &=T-L_1^{\min}\left(\frac{1}{R}-\frac{1}{c}\right)+\sup_{t \geq 0} \left\{\frac{w(t)}{R}-t\right\}.
\end{align}
By \eqref{eq:pkt-bit-1} and \eqref{eq:pkt-bit-2}:
\begin{align}
    \Delta^{pkt} - \Delta^{\mathrm{A}} = -(L_1^{\max} - L_1^{\min})\left(\frac{1}{R}-\frac{1}{c}\right) <0,
\end{align}
as $L_1^{\min}<L_1^{\max}$ and $R<c$.

\subsection{Proof of Proposition \ref{pro:compare-g-arr-delay}} \label{proof:compare-g-arr-delay}
\label{proof:compare-g-arr-delay}
(i)	 Theorem \ref{thm:response_fifo_arrival} gives:
	\begin{equation}\label{eq:compare-g-arr-1}
	    \Delta^\A(l) = h(\alpha^+ + {\alpha'}^+-l,\beta) + \frac{l}{c}.
	\end{equation}
	From item 1 of Proposition~\ref{pro:arr-g}, $g_1(x)=\alpha^{\downarrow}(x+L_1^{\min})$ and \mbox{$g_2(x)={\alpha'}^{\downarrow}(x+L_2^{\min})$}. By \eqref{eq:jylb_theory_1}, we have:
	\begin{align}
	    \nonumber g_1(x)=(\alpha^+)^{\downarrow}(x+L_1^{\min}),~g_2(x)=({\alpha'}^+)^{\downarrow}(x+L_2^{\min}).
	\end{align}
	
	Since $(\alpha^+)^{\downarrow}$ and $({\alpha'}^+)^{\downarrow}$ are left continuous and respectively larger than or equal to $L_1^{\max}$ and $L_2^{\max}$, by \eqref{eq:comp-up}:
	\begin{align}
	    \nonumber g_1^{\uparrow}(t) &= ((\alpha^+)^{\downarrow})^{\uparrow}(t) - L_1^{\min},\\
	     g_2^{\uparrow}(t) &= (({\alpha'}^+)^{\downarrow})^{\uparrow}(t) - L_2^{\min}.
	\end{align}
	As $\alpha^+$ and ${\alpha'}^+$ are right continuous, by \eqref{eq:liebeherr_lemma10.1d} we have:
	\begin{align}
	    g_1^{\uparrow}(t) = \alpha^+(t) - L_1^{\min},~g_2^{\uparrow}(t) = {\alpha'}^+(t) - L_2^{\min}.
	\end{align}
	Then, by applying Theorem \ref{thm:greg-delay}, we obtain 	
	\begin{align}\label{eq:compare-g-arr-2}
	    \Delta^\G(l)= h(\alpha^++ {\alpha'}^+-L_1^{\min}-L_2^{\min}+L_2^{\max},\beta) + \frac{l}{c}.
	\end{align}
	Let us compare \eqref{eq:compare-g-arr-1} and \eqref{eq:compare-g-arr-2}. Since $l\geq L_1^{\min}$ and \mbox{$L_2^{\max}\geq L_2^{\min}$}, we have:
	\begin{align}
	    \alpha^+ + {\alpha'}^+-l \leq \alpha^++ {\alpha'}^+-L_1^{\min}-L_2^{\min}+L_2^{\max}.
	\end{align}
	By \cite[Proposition~5.12]{bouillard_deterministic_2018}, $h$, is monotonically increasing with respect to its first argument; hence, $\Delta^A(l)\leq \Delta^G(l)$.
	
	Since $\Delta^A(l)\leq \Delta^G(l)$ holds for all packet sizes $l$, it also holds that $\sup_l \Delta^A(l)\leq \sup_l \Delta^G(l)$, i.e., $\Delta^A \leq \Delta^G$. 
	If  $L_1^{\max}=L_1^{\min}$ and $L_2^{\max}=L_2^{\min}$ (the packets of each flow have the same length), then $\Delta^A= \Delta^G$. For the general statement, we show a case that when $L_1^{\min}<L_1^{\max}$ or $L_2^{\min}<L_2^{\max}$, we have $\Delta^A < \Delta^\G$. 

First, for the ease of presentation, let us define \mbox{$w(t)=\alpha(t) + \alpha'(t)$}. Next, assume a rate-latency service-curve \mbox{$\beta(t)=R[t-T]^+$}, \mbox{$R< c$}.
Then, as \mbox{$\beta^{\downarrow}(x)=T+\frac{x}{R}$}, Theorem~\ref{thm:response_fifo_arrival} gives:
	\begin{align}\label{eq:arr-g-1}
    \nonumber\Delta^{\A} &= \sup_{t \geq 0} \left\{\beta^{\downarrow}\left(w(t)-L_1^{\min}\right)-t\right\}+\frac{L_1^{\min}}{c}\\
    \nonumber&=\sup_{t \geq 0} \left\{T+\frac{w(t)-L_1^{\min}}{R}-t\right\}+\frac{L_1^{\min}}{c}\\
    &=T-\frac{L_1^{\min}}{R}+\sup_{t \geq 0} \left\{\frac{w(t)}{R}-t\right\}+\frac{L_1^{\min}}{c}.
\end{align}
Using the derived g-regularity constraints of flows $1$ and $2$, Theorem~\ref{thm:greg-delay} gives:
\begin{align}\label{eq:arr-g-2}
    \nonumber\Delta^{\G} &= \sup_{t \geq 0} \left\{\beta^{\downarrow}\left(w(t)-L_1^{\min}-L_2^{\min}+L_2^{\max}\right)-t\right\}+\frac{L_1^{\max}}{c}\\
    \nonumber&=\sup_{t \geq 0} \left\{T+\frac{w(t)-L_1^{\min}-L_2^{\min}+L_2^{\max}}{R}-t\right\}+\frac{L_1^{\max}}{c}\\
    &=T-\frac{L_1^{\min}}{R}+\sup_{t \geq 0} \left\{\frac{w(t)}{R}-t\right\}+\frac{L_2^{\max}-L_2^{\min}}{R}+\frac{L_1^{\max}}{c}.
\end{align}
By \eqref{eq:arr-g-1} and \eqref{eq:arr-g-2}:
\begin{align}
    \Delta^{\A} - \Delta^\G = -\frac{L_2^{\max}-L_2^{\min}}{R} - \frac{L_1^{\max}-L_1^{\min}}{c} <0,
\end{align}
as $L_1^{\min}<L_1^{\max}$ or $L_2^{\min}<L_2^{\max}$.
		
(ii) Theorem \ref{thm:greg-delay} gives:
    \begin{align}
    	\Delta^\G(l)= h(g_1^{\uparrow} + g_2^{\uparrow} + L^{\max}_2,\beta) + \frac{l}{c}.
	\end{align}
    From item 2 of Proposition \ref{pro:arr-g}, we obtain $\alpha(t)=g_1^{\downarrow}(t) + L_1^{\max}$ and $\alpha'(t)=g_2^{\downarrow}(t) + L_2^{\max}$. Using Theorem \ref{thm:response_fifo_arrival}, we have:
	\begin{align}
	\nonumber \Delta^\A(l)&=h((g_1^{\downarrow})^+ + (g_2^{\downarrow})^+ + L_2^{\max}+L_1^{\max}-l,\beta) + \frac{l}{c}\\
	&=h(g_1^{\uparrow} + g_2^{\uparrow} + L_2^{\max}+L_1^{\max}-l,\beta) + \frac{l}{c}.
	\end{align}
	Note that by \eqref{eq:jylb_theory_2}, $(g_1^{\downarrow})^+=g_1^{\uparrow}$ and $(g_2^{\downarrow})^+=g_2^{\uparrow}$. 
    Similarly to the proof of (i), due to monotony of $h$ with respect to its first argument, $\Delta^\G(l) \leq \Delta^\A(l)$.
		
	Since $\Delta^G(l)\leq \Delta^A(l)$ holds for all packet sizes $l$, it also holds that $\sup_l \Delta^G(l)\leq \sup_l \Delta^A(l)$, i.e., $\Delta^G \leq \Delta^A$. 
	If  $L_1^{\max}=L_1^{\min}$ (the packets of the flow if interest have the same length), then $\Delta^A= \Delta^G$. For the general statement, similarly to the proof of item (i), we show a case that when $L_1^{\min}<L_1^{\max}$, we have $\Delta^A < \Delta^\G$. Considering the same service curve \mbox{$\beta(t)=R[t-T]^+$}, \mbox{$R< c$} as proof of item (1), we obtain:
	\begin{align}
	    \Delta^{\G} - \Delta^\A = -(L_1^\max-L_1^\min)\left(\frac{1}{R}-\frac{1}{c}\right) <0,
	\end{align}
    as $L_1^{\min}<L_1^{\max}$ and $R<c$.

\section{Example of non $c$-Lipschitz service curve}\label{appendix-example}
Consider the following function (FIFO residual service curve \cite{bouillard_deterministic_2018}), where $\theta$ and $R$ are fixed positive numbers:
$$\beta(t)=\begin{cases} 0 &\mbox{if } t \leq \theta,\\ 
	R t & \mbox{if }t > \theta.
\end{cases}$$
It is not $c$-Lipschitz as it is not continuous at $t=\theta$. Considering the assumptions of item (i) of Theorem~\ref{thm:response_fifo_arrival}, the response time of a packet with size $l$ of flow $1$ is 
\begin{align}
	\Delta^\A(l)= \sup_{t\geq 0} \left\{\beta^{\downarrow}\left(\alpha^+(t)+{\alpha'}^+(t)-l\right)-t   \right\}
	+\frac{l}{c}.
\end{align}
and the delay bound for flow $1$ is:
\begin{align}
    \Delta^\A = \sup_{l\in[L^\min_1,L^\max_1]} \left\{\Delta^\A(l)\right\}
\end{align}
Given that $\alpha^+(t) \geq L_1^{\max}$, ${\alpha'}^+(t) \geq L_2^{\max}$ and $l\leq L_1^{\max}$, we have
\begin{align}
    \beta^{\downarrow}\left(\alpha^+(t)+{\alpha'}^+(t)-l\right)=\max(\theta, \frac{\alpha^+(t)+{\alpha'}^+(t)-l}{R})
\end{align}
Therefore,
\begin{align}
    \nonumber \Delta&^\A(l) =\sup_{t\geq 0} \left\{\max(\theta, \frac{\alpha^+(t)+{\alpha'}^+(t)-l}{R})-t   \right\}
	+\frac{l}{c}\\
	\nonumber &=\sup_{t\geq 0} \left\{\max(\theta - t, \frac{\alpha^+(t)+{\alpha'}^+(t)-l}{R}-t)   \right\}
	+\frac{l}{c}\\
	\nonumber &=\max\left(\sup_{t\geq 0}\left\{\theta - t\right\},\sup_{t\geq 0}\{\frac{\alpha^+(t)+{\alpha'}^+(t)-l}{R}-t\}\right)+\frac{l}{c}\\
	&=\max\left(\theta+\frac{l}{c},\psi-\frac{l}{R}+\frac{l}{c}\right),
\end{align}
with 
$$
\psi = \sup_{t\geq 0}\left\{\frac{\alpha^+(t)+{\alpha'}^+(t)}{R}-t\right\}.
$$
After examining all cases and some algebra, we find that 

\begin{itemize}
    
    \item if $R(\psi - \theta) \leq (1-\frac{R}{c})L^\min_1+\frac{R}{c}L^\max_1$ then 
    \begin{align}
    \nonumber \Delta^\A &=\sup_{l\in[L^\min_1,L^\max_1]} \left\{\Delta^\A(l)\right\}=\theta+\frac{L^\max_1}{c}
    \end{align}
    and the supremum is attained at $l=L^\max_1$;
    
    
    
    \item else \begin{align}
        \nonumber \Delta^\A &= \psi-\frac{L^\min_1}{R}+\frac{L^\min_1}{c},
    \end{align} and the supremum is attained at $l=L^\min_1$ and not at $l=L^\max_1$.
    
\end{itemize}
Therefore, the supremum over $[L^\min_1,L^\max_1]$ to obtain $\Delta^\A$ can be achieved either at $l=L^\max_1$ or $l=L^\min_1$, depending on the parameter values.





\end{document}